\documentclass[nonacm]{acmart}
\makeatletter                   
\def\mdseries@tt{m}             
\makeatother                    
\usepackage[plain]{fancyref}
\usepackage[draft=true]{minted} 
\usepackage{color}
\usepackage{hyperref}           
\hypersetup{
    colorlinks=true,
    linkcolor=blue,
    filecolor=red,      
    urlcolor=magenta,
    breaklinks=true,            
}
\usepackage{breakurl}           

\def\BibTeX{{\rm B\kern-.05em{\sc i\kern-.025em b}\kern-.08emT\kern-.1667em\lower.7ex\hbox{E}\kern-.125emX}}

\newcommand{\TRvsCCSC}[2]{{#1}}   
\newcommand{\code}[1]{\texttt{#1}} 
\newcommand{\newterm}[1]{\textit{#1}}
\newcommand{\unused}[1]{}

\newcommand{\nobracket}{}
\newcommand{\tmtextbf}[1]{{\bfseries{#1}}}
\newcommand{\tmtextit}[1]{{\itshape{#1}}}
\newcommand{\tmverbatim}[1]{{\ttfamily{#1}}}

\newcommand{\textdots}[0]{\tmverbatim{...}}
\newcommand{\dwln}[0]{\cline{1-1}}
\newcommand{\dwupl}[0]{\vspace{-3mm}}
\newcommand{\dwupr}[0]{\vspace{-4mm}}

\usepackage{tabu}

\newcommand{\bnfdef}{%
  \mathrel{\!\colon\!\!\colon\!\mathord{=}}%
}

\usepackage[english]{babel}

\usepackage{mathpartir}

\usepackage{amsthm}
\newtheorem{claim}{Claim}[section]
\newtheorem*{claim*}{Claim}

\newminted{ocaml}{}
\newmintinline{ocaml}{}
\newcommand{\oi}[1]{\ocamlinline{#1}}

\newminted{coq}{}
\newmintinline{coq}{}
\newcommand{\coqi}[1]{\coqinline{#1}}

\newminted{c}{}
\newmintinline{c}{}
\newcommand{\ci}[1]{\cinline{#1}}

\newminted{python}{}
\newmintinline{python}{}
\newcommand{\pyi}[1]{\pythoninline{#1}}

\newminted{haskell}{}
\newmintinline{haskell}{}
\newcommand{\haski}[1]{\haskellinline{#1}}

\usepackage{xspace}
\newcommand{\eg}{\emph{e.g.\@\xspace}}
\newcommand{\ie}{\emph{i.e.\@\xspace}}
\newcommand{\etc}{\emph{etc.\@\xspace}}

\newcommand{\orca}{\textsc{Orca}\xspace}

\begin{document}
\TRvsCCSC{
\typeout{======================================}
\typeout{======== TECH REPORT VERSION =========}
\typeout{======================================}
}{
\typeout{======================================}
\typeout{============ CCSC VERSION ============}
\typeout{======================================}
}

%

\title{A Bridge Anchored on Both Sides}
\subtitle{Formal Deduction in Introductory CS, and Code Proofs in Discrete Math
}
%

%
\author{David G. Wonnacott}
\email{davew@cs.haverford.edu}
\affiliation{%
  \institution{Haverford College}
  \streetaddress{}
  \city{Haverford}
  \state{Pennsylvania}
  \postcode{19041}
}
\author{Peter-Michael Osera}
\email{osera@cs.grinnell.edu}
\affiliation{%
  \institution{Grinnell College}
  \streetaddress{}
  \city{Grinnell}
  \state{Iowa}
  \postcode{50112}
}

%

%
\begin{abstract}
There is a sharp disconnect between the programming and mathematical portions of the standard undergraduate computer science curriculum,
leading to student misunderstanding about how the two are related.
We propose connecting the subjects early in the curriculum---specifically, in CS1 and the introductory discrete mathematics course---by using formal reasoning about programs as a bridge between them.

This article reports on Haverford and Grinnell College's experience in constructing the end points of this bridge between CS1 and discrete mathematics.
Haverford's long-standing ``3-2-1'' curriculum introduces code reasoning in conjunction with introductory programming concepts, and Grinnell's discrete mathematics introduces code reasoning as a motivation for logic and formal deduction.
Both courses present code reasoning in a style based on symbolic code execution techniques from the programming language community, but tuned to address the particulars of each course.

These courses rely primarily on traditional means of proof authoring with pen-and-paper.
This is unsatisfactory for students who receive no feedback until grading on their work and instructors who must shoulder the burden of interpreting students' proofs and giving useful feedback.
To this end, we also describe the current state of \orca, an in-development proof assistant for undergraduate education that we are developing to address these issues in our courses.

\TRvsCCSC{%
  Finally, in teaching our courses, we have discovered a number of educational research questions about the effectiveness of code reasoning in bridging the gap between programming and mathematics, and the ability of tools like \orca to support this pedagogy.
  We pose these research questions as next steps to formalize our initial experiences in our courses with the hope of eventually generalizing our approaches for wider adoption.
}{}

\end{abstract}

%
%


%
\keywords{CS Education, CS1, deduction, discrete mathematics, proof, verification}

%

%
\maketitle

\section{Introduction}
\label{sec:introduction}

Mathematics is a critical part of computer science education,
as it is both essential to many aspects of practical system design
and a barrier to success for many students.
The latest revision of the ACM's Computer Science Curricula recognizes proofs as
one of several areas of mathematics that are integral to a wide variety of sub-fields of computer science.
According to the curricular guidelines:
\begin{quote}
\ldots an ability to create and understand a proof---either a formal symbolic proof or a less formal but still mathematically rigorous argument---is important in virtually every area of computer science, including (to name just a few) formal specification, verification, databases, and cryptography.~\cite{cscurricula:2013}
\end{quote}
This list is hardly exhaustive; for example,
computer security efforts have made significant headway though the use of formal methods~\cite{Quanta:2016hackerproof}.
Perhaps more importantly,
familiarity with formal proof can also support informal work,
by providing a richer understanding of code
and by serving as a tool for communication about algorithmic ideas~\cite{pierce:sf:2015,manber:algorithms:1989}.

Although faculty identify proofs as a key component of a computer science education~\cite{osera:sigcse:2014},
students and professionals frequently do not appreciate the relevance of mathematical proof to their careers~\cite{exter-2012}.
One problem is that students fail to grasp the \emph{mechanics} or rules of proof during the course of their studies.
Proof by induction seems to be particularly difficult for students~\cite{polycarpou:sigcse:2008, baker:thesis:1995, dubinsky:jmb:1986, dubinsky:jmb:1989},
though the proofs themselves are not inherently more complex than
the algorithms that these same students master.
This difficulty is exacerbated by the fact that students typically develop proof skills in a environment that differs greatly from the one in which they hone their coding skills.
Coding takes place in a feedback-rich environment,
which invites even a novice to explore new ideas and develop their own understanding through a rich variety of experiences.
In contrast, proof techniques are typically explained and demonstrated by a teacher,
and the student is expected to absorb the concepts before exploring them.
They then attempt to follow the model with pencil and paper,
receiving external feedback much later when the work is graded.

Educators have been concerned with the importance of exploration and feedback
since at least the work of Kolb on Experiential Learning in the 1980's,
if not all the way back to Dewey's Pragmatism decades earlier~\cite[Chapter 3]{gravells:ed:2014}.
This principle has not escaped teachers of mathematics,
and they have explored solutions to this problem in the context of their courses.
In-class supervised problem solving greatly reduces the feedback lag,
but it does not approach the rates possible with software.
While some logic/mathematics courses have responded to this challenge by introducing software,
for example the ``Logic and Proofs'' course at Carnegie-Mellon University~\cite{SS:ile:1994},
this addresses only one end of the proof/computation connection,
leaving unrealized the potential to motivate proof through its relevance in coding courses.

Few of the courses that introduce programming illustrate the relevance of deduction or proof,
even when the curriculum is clearly informed by a formal foundation.
For example, the ``How to Design Programs'' curriculum~\cite{HtDP} (HtDP, hereafter)
describes computation as a process of substitution,
essentially a special case of the substitutions made in a deductive proof.
However, while the text occasionally mentions that a particular property may be proved,
it leaves all actual proofs for ``a follow-up course''.
The term deduction appears (in the online draft of the 2nd edition) only in the context of a program making a payroll deduction.

The substitution process used to illustrate computation in this curriculum (and a few others)
relies on the same ``substitution of equals for equals'' principle that underlies deductive proof.
In high-school algebra, a set of definitions such as $a=14$, $b=7$,
$x = 2a+2b$ provides 
a variety of opportunities.
We can, of course, rewrite the definition of $x$ as $x=42$,
and then replace other $a$'s, $b$'s, and $x$'s with the corresponding numerical values.
Alternatively, we could replace occurrences of $x$ with $(2a+2b)$,
or with the equivalent $2(a+b)$.

In program code, variable definitions take on a different notation, such as
C++'s \ci{int x=2*a+2*b;} or
Python's \pyi{x=2*a+2*b},
usually require that all uses of a variable occur after its definition,
and can have very different meanings if later assignments modify the value associated with the name.
However, in the absence of subsequent assignments,
we (or an optimizing compiler) can also apply any of the above substitutions
to properly ordered/structured code
without changing a program's result.
\unused{
    This fact is often omitted from introductory curricula,
    resulting in considerable surprise (and occasional disbelief) when
    students begin to learn about the actual process of executing code
    in upper-level systems courses
    (\eg, on hardware, programming-languages, or concurrent programming).
    In our experience,
    beginners are quite capable of accepting this fact,
    especially if it is related to other examples in which
    they see several approaches to achieving one goal, \eg sorting.
}

When all code is \emph{pure}, \ie, free of side-effects,
substitution rules can also be used to describe
other programming language elements such as
\code{if/else}, function calls, \etc.
While it has been known for some time that this approach can be
generalized to produce a universal substitution-based view of computation~\cite{turing_1937},
the full mathematical treatment is not generally considered appropriate in the context of introductory computer science.
However, the underlying principle can enrich a new student's understanding of computation.

In their first exposure to programming,
a student begins to understand both the notation and conventions for expressing a program,
and what it means to execute a program.
This mental model of the process of computation is called a \newterm{notional machine}~\cite{DOM:notional:1999}.
It may be taught explicitly, or the student may develop it implicitly during activities where they are forced to think about how programs execute, such as debugging.
Studies of student outcomes argue for explicit teaching,
and plenty of practice with the notional machine, as
``changing an ingrained but flawed mental model is more difficult than
helping a model to be constructed in the first place''\cite{Sorva:2013:NMI}.

Notional machines of introductory programming courses often resemble a typical debugger:
program control moves from line to line;
as it does so, the machine updates the set of variables,
which are shown as a group of name-value pairs (or name-object-value relationships)
for each program scope.
This \newterm{dictionary notional machine}~\cite{tfk:sigcse:2018notional}
may be taught formally or informally, or simply noticed by the students as they debug.

The HtDP curriculum begins with a pure-functional language and an
explicitly-presented \newterm{substitution notional machine}~\cite{tfk:sigcse:2018notional}
called \newterm{Beginning Student Language} (or \newterm{BSL}, for short).
HtDP and BSL have been used both in the classroom
and recently for studies of the benefits of explicit teaching of BSL~\cite{fkt:sigcse:2017sma,tfk:sigcse:2018notional}.

Haverford's ``3-2-1'' introductory CS curriculum~\cite{DW:sigcse:2005}
also provides early and explicitly focus on notional machines.
This curriculum differs from HtDP in a number of ways.
In particular, the differences between BSL and
the 3-2-1 substitution notional machine (referred to as \newterm{the 3-2-1 SNM} hereafter)
make the latter more suitable for bridging to proof.
Specifically, the 3-2-1 SNM can demonstrate, in a single framework, both
\textit{execution} of code for a specific set of inputs
and \textit{deduction} of an abstract property of the code,
\ie, a property that is true for any input.
The 3-2-1 curriculum also contains formal verification examples that are presented
in the structure and language of mathematical proof,
though these are not necessary for the primary theme of connecting deduction and execution.

Note that the 3-2-1 curriculum lacks a number of appealing features of HtDP,
such as a complete textbook,
formally-defined execution models for all its notional machines,
and, until recently, tool support
(HtDP's execution models have strong support in the DrRacket environment).
However, the \orca software tool~\cite{osera-blocks-2017} can be configured to support
the 3-2-1 curriculum,
and the development of this tool has spurred a re-examination of the core 3-2-1 material,
to take advantage of \orca and streamline the pedagogy in other ways.

\orca is also influencing the development of the discrete mathematics curriculum at Grinnell College, which uses code reasoning to motivate logic and deductive reasoning.
In this context, students are already acquainted with introductory programming concepts but are (a) skeptical of the utility of formal reasoning in programming and (b) unsure of their own abilities as mathematicians.
Code reasoning in discrete mathematics gives students a concrete connection between what they already know---programming---and what they are about to learn---formal logic and reasoning, and hopefully boosts their confidence in tackling the material.


This article presents the current state of our work to use these two courses to
build a pedagogical bridge between CS1 and discrete math,
with both courses emphasizing the relevance of formal deduction and proof in thinking about computation.
We have a coherent curriculum that introduces computation in CS1
(see \autoref{sec:substitution})
in a manner that supports deductive reasoning and all elements of direct and indirect proof
(\autoref{sec:symbolic-exec}),
and similarly introduces proofs as a way to reason about code, in discrete mathematics (\autoref{sec:discrete-mathematics}).
These courses anchor both ends of the pedagogical bridge between coding and mathematics,
and this connection is made more evident by our use of a single software tool,
\orca (\autoref{sec:orca}),
and a shared set of \orca-supported examples for both courses (\autoref{sec:integrating}).
Our informal observations about student responses to earlier versions of this curriculum
are included throughout the aforementioned sections,
but we have not (yet) performed any formal field studies of this curriculum.
\TRvsCCSC{
We present a list of questions that we would like to answer about the benefits of this curriculum (\autoref{sec:experiments}),
though actual studies are left for our future work.
}{}

The framing, and the associated curricula, are novel in a number of ways:
\begin{itemize}
    \item they are presented with a standard high-demand programming language, without requiring departure from common idioms;
    \item they connect computation, deduction, and proof more directly than other curricula such as HtDP,
          while being conceptually simpler than the previous 3-2-1 curriculum
          (we've had to work to avoid accidental over-simplification through tools that automate too much of the work);
    \item they allow the introduction of standard mathematical conjecture/proof structure and terminology
          on the same example proofs that were framed as symbolic substitution-based execution,
          in either in CS1 or discrete math;
    \item they are built around a single software tool, \orca,
          which supports both courses,
          and has led to a number of pedagogical simplifications in each.
\end{itemize}

\unused{

For other aspects of the 3-2-1 curriculum,
refer to~\citeN{DW:sigcse:2005}
or the course notes for the 3-2-1 curriculum's first semester~\cite{Wonnacott:fvte17}.
For other details about HtDP,
refer to~\citeN{HtDP},
the associated \texttt{https://htdp.org} web site,
or studies of this curriculum such as~\cite{fkt:sigcse:2017sma,tfk:sigcse:2018notional}.
}

\section{Language and Substitutional Execution in the 3-2-1 Curriculum}
\label{sec:substitution}

The 3-2-1 curriculum begins with a minimal set of programming language features
and explains the use of the its notional machines for those features.
After exploring software design and execution within that minimal set,
it moves on to explore a variety of other topics, including additional language features as needed.
Here, we summarize the language features of this initial subset,
the relevant elements of the 3-2-1 SNM,
and the pedagogical issues surrounding their use to introduce deductive reasoning about code.
The details below differ from the most recent (2017) course notes~\cite{Wonnacott:fvte17}
in several ways,
including formatting,
choice of example,
and, most significantly, some of the specifics of the substitution rules.
We plan to integrate this newer treatment into the next edition.
Other aspects of 3-2-1's first semester,
such as
testing methods and test-suite development,
imperative programming,
code-reading and code review processes,
and complexity analysis,
are not currently undergoing rapid change;
refer to~\citeN{DW:sigcse:2005} or ~\cite{Wonnacott:fvte17} for more details about them.

Haverford has recently moved to using a subset of Python 3.6 for CS1
and Java for CS2 (the latter introduces class definitions and inheritance/interfaces).
Previous versions of the 3-2-1 curriculum had employed a single language for the full first year,
specifically Python 2, Python 3.6, or C++.
The main language requirements for the first semester are
support for multiple paradigms (imperative, pure-functional, and object-oriented features of library objects),
statically-scoped local variables,
and a conditional expression (such as the \ci{? :} operation in C++ and Java).
\unused{
    We also have a strong preference for static checking,
    to reduce both debugging time and questions about why we manually check some things but not others.
}
Regardless of language choice,
we employ a pure-functional subset in the first half of the first semester,
so that the 3-2-1 SNM and the informal dictionary notional machine
can be presented with equal ease.

This curriculum could also be employed with a language that places a greater emphasis on pure-functional style,
such as Haskell, ML, or Scheme/Racket.
However, our goal is to demonstrate the thinking techniques in the student's comfort/interest zone,
which is often defined as ``a language that will get me a job''
or ``a language I've heard about from a friend'',
rather than ``a language my teacher thinks is cool''.

The first semester's path from execution to deduction relies on language and SNM features
for integer, (immutable) string, and Boolean values,
all manipulated via functions containing
\code{return} statements and potentially
\code{if} statements/expressions or local variables.

\subsection{Arithmetic and Algebra on Integers}
\label{sec:subst-arith}

The 3-2-1 curriculum introduces program execution as a problem that can be solved in several ways,
illustrating the substitution and dictionary notional machines, in turn,
on specific code examples such as the code in Figure~\ref{fig:ints-code}.
\begin{figure}[tb]
  \input{CodeInts/exec_2a_2b_code.tex}
\caption{Code for an early example of execution.}
\label{fig:ints-code}
\end{figure}
\begin{figure}[tb]
  \centering 
  \includegraphics[width=\linewidth]{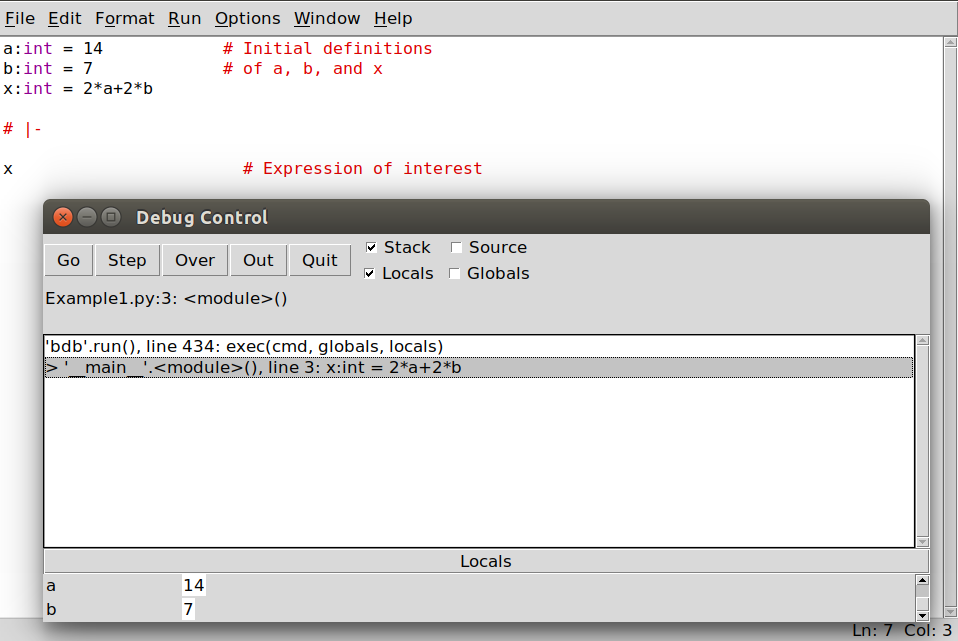}
\caption{Executing figure~\ref{fig:ints-code} in the Idle debugger.}
\label{fig:ints-idle}
\end{figure}
\begin{figure}[tb]
  \centering
  \includegraphics[width=\linewidth]{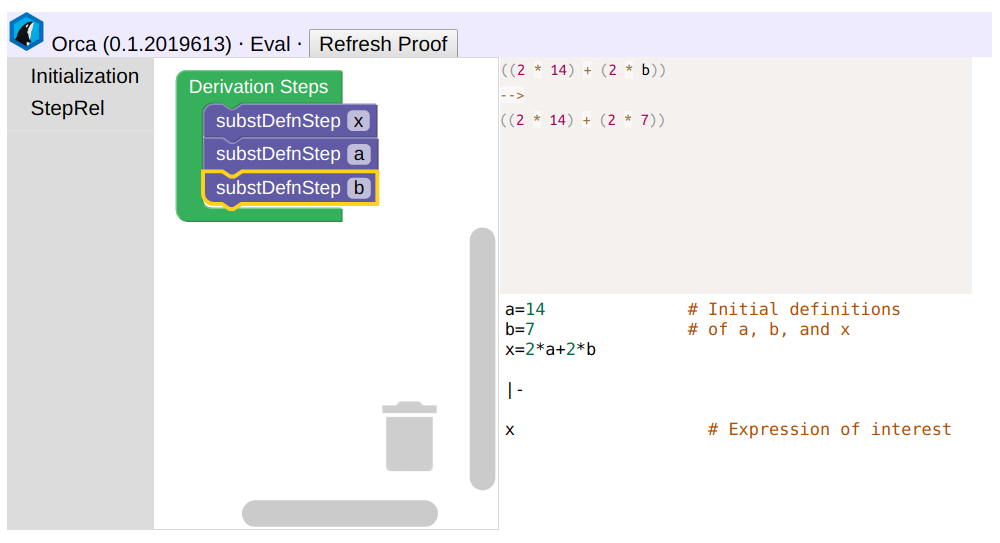}
\caption{Executing figure~\ref{fig:ints-code} in \orca.}
\label{fig:ints-orca}
\end{figure}
The dictionary notional machine is illustrated by single-stepping Figure~\ref{fig:ints-code} with a standard debugger.
Figure~\ref{fig:ints-idle} shows a debugging session with the debugger of Idle Python IDE; 
as we are about to execute line 3, \pyi{a} and \pyi{b} have been defined,
but \pyi{x} has not.
The 3-2-1 SNM can be illustrated by the instructor at the board, if necessary,
or with \orca, as in Figure~\ref{fig:ints-orca}.
The central pane of the \orca window shows the sequence of substitutions;
the user can use the mouse to select and drag substitution rules,
to add, remove, or reorder this sequence.
\orca's lower-right pane shows the definitions,
and the upper right shows the expression before and after
application of the highlighted rule (in this case, substituting the definition of \pyi{b}).
For later reference, 
we record the progress of this machine with a two-column sequence
such as the one shown in Figure~\ref{fig:ints-arith}.
(Note that some of the names of substitution rules
have not been updated in \orca yet;
SubstDefnStep will be renamed name-to-def in the next update,
for consistency with the terminology below.)

Note that Figure~\ref{fig:ints-code},
and the relevant examples in our curriculum (and below),
are formatted to work with \orca,
\ie, as a sequence of definitions followed by a commented turnstile (\pyi{|-}) mark and then an expression of interest. 
\begin{figure}[tb]
  \centering
  
\begin{tabular}{l|lp{3.0cm}}
  \tmverbatim{\begin{tabular}{l}
    x
  \end{tabular}} & \quad & \ \\
\dwln \dwupl & & \dwupr name-to-def (\tmverbatim{x})\\
  \tmverbatim{\begin{tabular}{l}
    2*a+2*b
  \end{tabular}} &  & \ \\
\dwln \dwupl & & \dwupr name-to-def (\tmverbatim{a})\\
  \tmverbatim{\begin{tabular}{l}
    2*14+2*b
  \end{tabular}} &  & \ \\
\dwln \dwupl & & \dwupr name-to-def (\tmverbatim{}\tmverbatim{b})\\
  \tmverbatim{\begin{tabular}{l}
    2*14+2*7
  \end{tabular}} &  & \ \\
\dwln \dwupl & & \dwupr arithmetic\\
  \tmverbatim{}\tmverbatim{}\tmverbatim{\begin{tabular}{l}
    42
  \end{tabular}} &  & \ 
\end{tabular}

\

\caption{Complete substitution sequence for figure~\ref{fig:ints-code}}
\label{fig:ints-arith}
\end{figure}
The sequence of substitutions starts with the expression below
the turnstile (which is detected automatically by \orca),
and proceeds by application of substitution rules (in the right column).

For Figure~\ref{fig:ints-arith},
we need only two substitution rules:
\begin{description}
    \item[arithmetic] Use the rules of arithmetic operations (e.g., \pyi{+} and \pyi{*}),
                      including operator precedence,
                      to produce an equivalent value or expression.
    \item[name-to-def]  For a variable defined \pyi{var = expr},
                      replace occurrences of \pyi{var} with \pyi{expr}.
                      Note that we only apply the 3-2-1 SNM
                      to code written in a ``pure'' subset of Python
                      (see Section~\ref{sec:subst-body}),
                      to ensure this process produces the same result as the standard Python interpreter.
                      Local variable definitions can be removed once
                      and the variable is no longer used
                      (\eg, after the uses have been replaced by this rule)
                      and the defining expression is \textit{safe}
                      (\eg, does not contain a term like \pyi{a/b} for which \pyi{b} might be 0;
                      see Section~\ref{sec:subst-vs-algebra} for details).
\end{description}
When the 3-2-1 SNM is executed by hand,
the rules are selected from the course notes (or, on exams, from a provided summary).
Application of these rules,
like any other form of program execution,
is easier and less-error prone when performed by a machine,
\eg, with \orca.

The 3-2-1 SNM
gives the user a number of choices during execution,
including the order in which rules are applied.
For example, the code of \ref{fig:ints-code} could also be executed in the order shown in Figure~\ref{fig:ints-arith2},
in which we have chosen to wait as long as possible to substitute the definition of \pyi{a}.
\begin{figure}[tb]
  \centering
  
\begin{tabular}{l|lp{3.0cm}}
  \tmverbatim{\begin{tabular}{l}
    x
  \end{tabular}} & \quad & \ \\
\dwln \dwupl & & \dwupr name-to-def (\tmverbatim{x})\\
  \tmverbatim{\begin{tabular}{l}
    2*a+2*b
  \end{tabular}} &  & \ \\
\dwln \dwupl & & \dwupr name-to-def (\tmverbatim{b})\\
  \tmverbatim{\begin{tabular}{l}
    2*a+2*7
  \end{tabular}} &  & \ \\
\dwln \dwupl & & \dwupr arithmetic\\
  \tmverbatim{\begin{tabular}{l}
    2*a+14
  \end{tabular}} &  & \ \\
\dwln \dwupl & & \dwupr name-to-def (\tmverbatim{a})\\
  \tmverbatim{\begin{tabular}{l}
    2*14+14
  \end{tabular}} &  & \ \\
\dwln \dwupl & & \dwupr arithmetic\\
  \tmverbatim{}\tmverbatim{}\tmverbatim{\begin{tabular}{l}
    42
  \end{tabular}} &  & \ 
\end{tabular}

\

\caption{Substituting a definition later for figure~\ref{fig:ints-code}.}
\label{fig:ints-arith2}
\end{figure}
The ``arithmetic'' rule is usually used in its default mode, to produce a single number,
but the user can request other equivalent expressions,
\eg, turning \pyi{2*14} into \pyi{4*7} or \pyi{14+14}.
Note that it does not allow illegal steps,
\eg, replacing the \pyi{14+2} with \pyi{16} in \pyi{2*14+2*7}.

\unused{
        *****************************************************************************
        *************************  Thinkning thimking thinkking  *********************
        *****************************************************************************
        
        Need a single place to discuss code we aren't scared of, and concise terms suggestive of what we mean.
        
        Scary things:
        * type errors (mypy)
        * non-universal builtins or trusted functions (x/y, power(x, y), with out substituton of y or 'if' on y)
        * un-trusted functions (e.g., power, before we've proved it)
        
        Note that this will involve information from the context (which will all eventually become if's)
        
        Do this discussion last, in 'comparison to algebra', but get the names right:
        
        Terms:
        * 'trusted': A function/operation is trusted if it's part of Python Language (+, **, /) or Library (math.pow) or we've verified it. This means that, if we call it \textit{and its precondition holds}, it will return an answer that fits the postcondition, neither running forever, throwing an exception, nor just giving the wrong answer
        * 'safe': Code/subexpression that will produce some answer, terminating without an exception, i.e., code is trusted and only asked to do things it says it can do
        + (informally) 'if they are appropriately safe', to hint (until 'discertionary' and 'comparison w/ algebra') about non-strict elements
        
        trust is about the operation/function (e.g., I trust my brother)
        but safety is in context: which is (currently) the conjunction of all if's conditions (I trust my brother, but it is not safe to be in a plane he's flying, because he does not know how to fly a plane) ... is "conjunction of if's" sufficient, or at least subsuming all things in course note?
        
        F17 notes state (P. 66, section 5.2.3) that we ``know''
        A Language rules (OK)
        B Rules for basic types (int, bool ... could include list and string, but F17 did not)
        C the precondition of the function being proved (now the outermost ``if'' ... should mention that vars there are assumed safe)
        D if's (mention here ... note that we're glossing over knowing x != 0 in the else of an 'if x==0')
        E "Rules from the step we're proving" (seems to me that C is part of this; E described in terms of progress) --> NOW OK, I THINK
        F convert the informal spec to formal, e.g. 'length of hypotenuse' to 'c s.t. c^2 = a^2 + b^2'

        note: in this course, only trusted things are safe, though in principle we could have something we know will terminate but we don't know that it will meet its postcondition, and might want to call that safe.

        *****************************************************************************
        *****************************************************************************
        *****************************************************************************
}
Of central importance to our deduction/proof pedagogy
is our ability to apply this rule to expressions that contain variables,
\eg, replacing \pyi{1+(x-1)} with \pyi{x}, or replacing either \pyi{x==x} or \pyi{x+y == y+x} with \pyi{True},
as long as the terms being manipulated are \textit{safe},
\eg, we can't simplify \pyi{x+y == y+x} to \pyi{True}
if \pyi{x} is a string and \pyi{y} a Boolean,
or if x is defined with an unsafe expression such as \pyi{a/b}, where \pyi{b} could be 0.
For details about safety, see Section~\ref{sec:subst-vs-algebra}.

\subsection{Rules for Booleans and Strings}
\label{sec:subst-if}
\begin{figure}[tb]
  \centering   
  \input{CodePunct/exec_punctuate1.tex}
\caption{Recommending punctuation, statement form.}
\label{fig:punct1-code}
\end{figure}\begin{figure}[tb]
  \centering   
  \input{CodePunct/exec_punctuate1f.tex}
\caption{Recommending punctuation, expression form.}
\label{fig:punct1f-code}
\end{figure}

Substitution rules for string expressions
can be presented to the students with varying levels of formality;
we tend to use informal language and references to Python documentation,
\eg, saying ``when \pyi{s1} and \pyi{s2} are strings,
\pyi{s1+s2} can be replaced with a string of all the symbols from \pyi{s1} followed by the symbols from \pyi{s2}.''
Our course examples do not currently require the simplification of expressions
that still contain un-evaluated string variables.
In principle, these could be included if we added axioms and supported them in \orca.

Boolean expressions can be simplified via the usual rules of Boolean algebra,
assuming the terms are specific Boolean values or otherwise appropriately safe.
We discuss both the Python's \pyi{if} statement and \pyi{if} expression,
though typically we present code in the former,
as in Figure~\ref{fig:punct1-code},
since this seems to be more common in Python usage.
However, students are allowed to use either form in their code,
as both seem to be within the idioms accepted by Python programmers
(especially for simple expressions).
When \pyi{if} is used in statement form,
we restrict it to cases that can easily be converted to expressions,
and switch to the expression form before performing other substitutions.
This step is automatic in \orca.

Converting \pyi{if}'s to expression form lets us
express all intermediate substitution steps as legal Python code
without requiring even less idiomatic forms
such as applications of 0-argument anonymous functions containing \pyi{if} statements.
For example, Figure~\ref{fig:punct1f-code} shows an expression-form variant of Figure~\ref{fig:punct1-code}.
Our restrictions on \pyi{if}s also rule out code with certain mistakes,
such as failure to return a value in some branches of an \pyi{if/else},
or use of a variable that is not always initialized;
we hope to deploy the checker developed for \orca
along with the \code{mypy} type checker,
to catch such errors before student code is even run.
%
%

Once all \pyi{if} statements have been converted to expressions,
all necessary substitutions can be handled by the following three rules:
\begin{description}
    \item[if-True]  An expression of the form $(X$ \pyi{if True else} $Y)$ can be replaced with $(X)$
    \item[if-False] An expression of the form $(X$ \pyi{if False else} $Y)$ can be replaced with $(Y)$
    \item[if-irrelevant] An expression of the form $(X$ \pyi{if} $C$ \pyi{else} $X)$ can be replaced with $(X)$, regardless of whether $C$ is true or false.
\end{description}
These rules, like the rules of Section~\ref{sec:subst-arith},
can be applied despite the presence of variables,
so \pyi{(x if True else e)} can be turned into \pyi{x} without knowing precisely what expression \pyi{e} is.
Unlike other rules, they allow some unsafe expressions: only the $C$ term need be safe, above.
An \pyi{if} expression can also implicitly define a value for a variable,
\eg, \pyi{a} can be replaced with \pyi{b+c} in the $X$ clause of $(X$ \pyi{if a==b+c else} $Y)$.

\subsection{User-defined Functions and Local Variables}
\label{sec:subst-body}
\begin{figure}[tb]
  \centering
  
\begin{tabular}{l|lp{3.0cm}}
  \tmverbatim{\begin{tabular}{l}
    recPunct("What is it")
  \end{tabular}} & \quad & \ \\
\dwln \dwupl & & \dwupr name-to-body\\
  \tmverbatim{\begin{tabular}{ll}
    ( & "What is it"+'?'\\
    if & \tmtextit{"What is it"[0:4]}=='What' else\\
    & "What is it"+'.'\\
    ) & 
  \end{tabular}} &  & \ \\
\dwln \dwupl & & \dwupr string-arithmetic: \tmverbatim{[ ]}\\
  \tmverbatim{\begin{tabular}{ll}
    ( & "What is it"+'?'\\
    if & \tmtextit{\tmtextbf{'What'}=='What'} else\\
    & "What is it"+'.'\\
    ) & 
  \end{tabular}} &  & \ \\
\dwln \dwupl & & \dwupr string-arithmetic: \tmverbatim{==}\\
  \tmverbatim{\begin{tabular}{ll}
    ( & "What is it"+'?'\\
    if & \tmtextbf{True} else\\
    & \tmtextit{"What is it"+'.'}\\
    ) & 
  \end{tabular}} &  & \ \\
\dwln \dwupl & & \dwupr if-True\\
  \tmtextit{\tmverbatim{\begin{tabular}{ll}
    "What is it"+( & \tmtextbf{\tmtextit{'?'}}
  \end{tabular} )}} &  & \ \\
\dwln \dwupl & & \dwupr string-arithmetic: \tmverbatim{+}\\
  \tmverbatim{}\tmverbatim{}\tmtextbf{\tmverbatim{\begin{tabular}{l}
    "What it it?"
  \end{tabular}}} &  & \ 
\end{tabular}

\

\caption{Function calls, conditionals, strings.}
\label{fig:punct1}
\end{figure}
During the 3-2-1 curriculum's introduction to pure-functional coding,
students write idiomatic Python,
but using a subset that lets us easily convert each function to a single
(possibly conditional) return statement.
Previous versions of this curriculum required students to perform several
substitutions for a single function call,
but we are moving to emphasize a single rule that
relies on this conversion to hide the statement/expression dichotomy:
\begin{description}
    \item[name-to-body] An expression of the form \pyi{func(argList)} can be replaced with
                      the expression returned by the function \pyi{func},
                      after replacing all parameter variables with appropriate argument expressions from \pyi{argList},
                      substituting-away all local variables with the name-to-def rule,
                      and converting \pyi{if}'s to expression form (as discussed in Section~\ref{sec:subst-if}).

\end{description}

As with the name-to-def rule, we can only apply this rule if
the expressions used to define local variables,
and the expressions for \pyi{argList}, are safe.
Figure~\ref{fig:punct1} illustrates the use of our name-to-body and if-True rules
for the code of Figure~\ref{fig:punct1-code}.
The sequence for the code of Figure~\ref{fig:punct1f-code}
would be identical, except for the factoring-out of the \pyi{"What is it" +} sub-expression.

In the context of hand-applied substitution rules,
the new name-to-body rule could possibly be more error-prone than the rules of our previous curriculum,
in which students did a larger number of simpler substitutions.
However, the new approach is conceptually simpler,
as the old version of name-to-body introduced a temporary structure that was not legal Python code,
and the new rule does not contain such structures,
or force students to consider the statement/expression distinction,
or depart from idiomatic Python variable definitions.
By using an \orca, we hope to gain the benefits of conceptual simplicity
without the potential errors of a single step that is too complex for hand-execution.
At the instructor's discretion,
more details can be explored once students are comfortable with the big picture
(see Section~\ref{sec:subst-details}).


\begin{figure}[tb]
 {
 \newcommand{\nt}[1]{\textit{\textbf{#1}}}
 \newcommand{\der}{{$\rightarrow$}}
 \begin{tabular}{rl}
      \nt{func} \der & \pyi{def } \textit{f}\pyi{(}\textit{paramList}\pyi{) -> } \textit{returnType}\pyi{ :}\nt{b} \\
         \nt{b} \der & \nt{a}$^*$ \nt{r} \\
         \nt{r} \der & \pyi{return} \nt{e} \\
                   | & \pyi{if} \nt{e}\pyi{:} \nt{b}$_1$ \pyi{else:} \nt{b}$_2$\\
         \nt{a} \der & $v$ \pyi{:}\textit{type} \pyi{=} \nt{e} \\
                | & \pyi{if} \nt{e}\pyi{:} 
                       $v$ \pyi{:}\textit{type} \pyi{=} \nt{e}$_1$
                    \pyi{else:}
                       $v$ \pyi{:}\textit{type} \pyi{=} \nt{e}$_2$ \\
         \nt{e} \der &     $v$
                  \;|\; $f$\pyi{(}\nt{e}$_1, ..., $\nt{e}$_k$\pyi{)}
                  \;|\; \nt{e}$_2$ \pyi{if} \nt{e}$_1$ \pyi{else} \nt{e}$_3$
                  \;|\; \nt{e}$_1$\pyi{(+)}\nt{e}$_2$
 \end{tabular}
 \\[3mm]
 \begin{flushleft}
 Where \textit{v} is a variable of the right type,
 \textit{f} is a function of the right type,
 \pyi{(+)} is a supported arithmetic/comparison/logical operation of the right type, and
 \textit{paramList} is a list of parameters, with types, in Python 3.6 notation.
 Note that, in the second \nt{a} rule,
 the two $v$'s must have the same variable name (and type),
 but otherwise no two definitions in a given function may have the same name,
 and no local variable may use the name of a global variable.
 Local variables can only be used by statements in the same indented block containing their \nt{a}
 (\ie, as in block-scoped languages).
\unused{
     Definitions inside an \pyi{if} (the second \nt{assign} rule)
     are merged with the list of definitions from the other side of the \pyi{if};
     the \pyi{if} and \pyi{else} clauses must define the same list of names
     (in the same order, with the same type),
     and those variables can be used by subsequent statements indented with the enclosing \pyi{if}.
     Except for conditional definitions allowed by the ``merge'' above,
     no name may have more than one definition in a given function body.
 }
 \end{flushleft}
 }
\caption{Function definition rules.}
\label{fig:decs-conds}
\end{figure}
For those curious about our restrictions on \pyi{if} statements and local variables,
Figure~\ref{fig:decs-conds} provides details.
Indentation, function argument/return types, etc.,
are as in Python 3.6, and not repeated in the figure.

\subsection{Calls to Trusted Functions}
For calls to library functions,
we typically don't have, or don't want to read, the code of the function's body.
However, just as we trust Python's rules for, \eg, arithmetic and \pyi{if},
we can also trust the rules for the Python's library.
\label{sec:subst-spec}
\begin{figure}[tb]
  \centering   
  \begin{center}
  \begin{minipage}{0.5\textwidth}
  \input{CodeInts/exec_pow_call.tex}
  \end{minipage}
  \end{center}
\caption{Code that calls a library function.}
\label{fig:pow-call}
\end{figure}
\begin{figure}[tb]
  \begin{center}
  \begin{minipage}{0.5\textwidth}
  
\begin{tabular}{l|lp{4.0cm}}
  \tmverbatim{\begin{tabular}{l}
    17+math.pow(a, b)
  \end{tabular}} & \quad & \ \\
\dwln \dwupl & & \dwupr name-to-def (\tmverbatim{a} and \tmverbatim{b})\\
  \tmverbatim{\begin{tabular}{l}
    17+math.pow(5, 2)
  \end{tabular}} &  & \ \\
\dwln \dwupl & & \dwupr name-to-spec (\tmverbatim{math.pow})\\
  \tmverbatim{ 17+}$5^2$ &  & \ \\
\dwln \dwupl & & \dwupr arithmetic\\
  \tmverbatim{}\tmverbatim{}\tmverbatim{\begin{tabular}{l}
    42.0
  \end{tabular}} &  & \ 
\end{tabular}

\

  \end{minipage}
  \end{center}
\caption{Executing the code of figure \ref{fig:pow-call}.}
\label{fig:exec-pow-call}
\end{figure}
So, for the code of Figure~\ref{fig:pow-call},
we can replace \pyi{pow(5, 2)} with $5^2$,
since the documentation of \pyi{pow} states ``Return x raised to the power y.''.
Figure~\ref{fig:exec-pow-call} illustrates these steps in the 3-2-1 SNM
(as before, a standard debugger can be used to illustrate the dictionary notional machine).
We are currently discussing the question of whether to express all
``expected values'' in Python notation vs. allowing other notation
(e.g., using \pyi{5**2} rather than $5^2$);
this may depend on an instructor's preference for discussing subtle
details of the types of, e.g. \pyi{math.pow} and \pyi{**} in Python.

\begin{figure*}[tb]
  \begin{center}
  \begin{minipage}{0.85\textwidth}
  
\begin{tabular}{l|lp{4.0cm}}
  \tmverbatim{ 17 + math.pow(a, b)} & \quad & \ \\
\dwln \dwupl & & \dwupr name-to-spec (\tmverbatim{math.pow})\\
  \tmverbatim{}\begin{tabular}{ll}
    \tmverbatim{17 + (} & \tmverbatim{a}$^{\mathtt{b}}$\\
    \tmverbatim{ \ \ \ \ if} & \tmverbatim{not(a<0 and not b.is\_integer())
    else}\\
    & \tmverbatim{ERROR}\\
    \tmverbatim{ \ \ \ \ )} & 
  \end{tabular} &  & \ \\
\dwln \dwupl & & \dwupr name-to-def (\tmverbatim{a} and \tmverbatim{b})\\
  \tmverbatim{}\tmverbatim{}\tmverbatim{}\begin{tabular}{ll}
    \tmverbatim{17 + (} & \tmverbatim{5}$^{\mathtt{2}}$\\
    \tmverbatim{ \ \ \ \ if} & \tmverbatim{not(5<0 and not 2.is\_integer())
    else}\\
    & \tmverbatim{ERROR}\\
    \tmverbatim{ \ \ \ \ )} & 
  \end{tabular} &  & \ \\
\dwln \dwupl & & \dwupr arithmetic\\
  \tmverbatim{}\tmverbatim{}\tmverbatim{}\tmverbatim{}\tmverbatim{}\begin{tabular}{ll}
    \tmverbatim{17 + (} & \tmverbatim{5}$^{\mathtt{2}}$\\
    \tmverbatim{ \ \ \ \ if} & \tmverbatim{not(False and not True) else}\\
    & \tmverbatim{ERROR}\\
    \tmverbatim{ \ \ \ \ )} & 
  \end{tabular} &  & \ \\
\dwln \dwupl & & \dwupr Boolean arithmetic\\
  \tmverbatim{}\tmverbatim{}\tmverbatim{}\tmverbatim{}\tmverbatim{}\tmverbatim{}\tmverbatim{}\begin{tabular}{ll}
    \tmverbatim{17 + (} & \tmverbatim{5}$^{\mathtt{2}}$\\
    \tmverbatim{ \ \ \ \ if} & \tmverbatim{True else}\\
    & \tmverbatim{ERROR}\\
    \tmverbatim{ \ \ \ \ )} & 
  \end{tabular} &  & \ \\
\dwln \dwupl & & \dwupr if-True\\
  \tmverbatim{}\tmverbatim{}\tmverbatim{\begin{tabular}{l}
    17 +
  \end{tabular}}\tmverbatim{5}$^{\mathtt{2}}$ &  & \ \\
\dwln \dwupl & & \dwupr arithmetic\\
  \tmverbatim{}\tmverbatim{}\tmverbatim{ 42} &  & \ 
\end{tabular}

\

  \end{minipage}
  \end{center}
\caption{Substituting variables late for figure \ref{fig:pow-call}.}
\label{fig:exec-pow-call-spec}
\end{figure*}
As in other situations, the order of substitutions for Figure~\ref{fig:exec-pow-call} is up to the user,
and many choices are possible.
In some cases, we may wish to substitute variable values after other steps,
as was done in Figure~\ref{fig:ints-arith2}.

When we apply a standard library function, or other trusted function,
to something other than specific value,
we must include Python's rules about the \textit{parameters}
as well as its rules about the \textit{result}.
For numeric \pyi{x} and \pyi{y}, the documentation cautions that,
when ``... x is negative, and y is not an integer then pow(x, y) is undefined''
(e.g., for $x=-1$ and $y=0.5$,
for which \pyi{math.pow} would need to compute $\sqrt{-1}$).
We express these rules about argument values via a second function-call rule:
\begin{description}
    \item[name-to-spec] An expression of the form \pyi{func(argList)},
                      for a function that is \textit{trusted} (\eg from a library)
                      can be replaced with\\
                      \pyi{   (} \textit{expected result} \pyi{if} \textit{argList is legal} \pyi{else ERROR )}\\
                      where \textit{expected result} is the result defined in the function's specification (\eg documentation),
                      and the \textit{argList is legal} test expresses the requirements for the argument values.
\end{description}
All elements of \pyi{argList} must be safe to apply name-to-spec (or name-to-body).
In software engineering terms,
the name-to-spec rule replaces the call with the value specified by the function's postcondition,
if and only if the precondition holds.
For variable values like $a=-1$ and $b=0.5$,
the sequence of substitutions for \pyi{math.pow(a, b)} would end up using our if-False rule to produce a final result of \pyi{ERROR},
indicating an error.
Note that we do not express type information in this (or other) rules,
as these are usually handled by other tools (\eg, \code{mypy}, for Python 3.6 code),
but the \pyi{math.py} function's second argument has type \pyi{float}---
it is only restricted to an integral value when the first argument is negative.

Note that this name-to-spec rule of the 3-2-1 SNM
can only be employed when the postcondition identifies a specific value,
but not with postconditions in other forms,
such as
``\pyi{sqrt(x)} returns \pyi{y} such that
  \pyi{y*y} is, within the limits of the floating-point approximation. equal to \pyi{x}''.
(See Section~\ref{sec:symbolic-additional} for more on this point.)
Examples for this section of the course are chosen to ensure that postconditions
either produce a result that can be computed some other way in Python (\eg, \pyi{pow(x, y) == x**y})
or correspond to some operation that can be handled by other axioms,
(\eg, \pyi{pow(x, y) ==}$x^y$, as above, for which we use mathematical rules about exponentiation).


\subsection{Recursion}
Recursion follows naturally as an application of the previously-discussed rules.
As our name-to-body rule removes parameters and local variables,
it will not introduce intermediate forms with multiple variables with the same name.
Figure~\ref{fig:power-exy} illustrates the execution of the recursive
\pyi{power} function of Figure~\ref{fig:power-code}.

\begin{figure}[tb]
  \centering
  \input{CodePower/exec_p_7_2_code.tex}
\caption{Recursive power function, call thereof.}
\label{fig:power-code}
\end{figure}
\begin{figure*}[tb]
  \centering
  
\begin{tabular}{l|lp{3.0cm}}
  \tmverbatim{\begin{tabular}{l}
    power(\tmtextit{x}, \tmtextit{y})
  \end{tabular}} & \quad & \ \\
\dwln \dwupl & & \dwupr name-to-def (\tmverbatim{x} and \tmverbatim{y})\\
  \tmverbatim{\begin{tabular}{l}
    power(\tmtextit{\tmtextbf{3+2}}, 2)
  \end{tabular}} &  & \ \\
\dwln \dwupl & & \dwupr arithmetic\\
  \tmverbatim{\begin{tabular}{l}
    \tmtextit{power(\tmtextbf{5}, 2)}
  \end{tabular}} &  & \ \\
\dwln \dwupl & & \dwupr name-to-body\\
  \tmverbatim{\begin{tabular}{ll}
    ( & 5\\
    if & \tmtextit{2 == 1} else\\
    & 5 * power(x, \tmtextit{2-1})\\
    ) & 
  \end{tabular}} &  & \ \\
\dwln \dwupl & & \dwupr arithmetic\\
  \tmverbatim{}\tmverbatim{}\tmverbatim{}\tmverbatim{\begin{tabular}{ll}
    ( & \tmtextit{5}\\
    \tmtextit{if} & \tmtextit{\tmtextbf{False} else}\\
    & 5 * power(5, \tmtextbf{1})\\
    ) & 
  \end{tabular}} &  & \ \\
\dwln \dwupl & & \dwupr if-False\\
  \tmverbatim{}\tmverbatim{}\tmverbatim{}\tmverbatim{}\tmverbatim{}\tmverbatim{}\tmverbatim{\begin{tabular}{l}
    5 * \tmtextit{power(5, 1)}
  \end{tabular}} &  & \ \\
\dwln \dwupl & & \dwupr name-to-body\\
  \tmverbatim{}\tmverbatim{}\tmverbatim{\begin{tabular}{ll}
    5 * ( & \tmtextbf{5}\\
    \ \ \ \tmtextbf{if} & \tmtextbf{\tmtextit{1 == 1} else}\\
    & \tmtextbf{5 * power(x, 1-1)}\\
    \ \ \ ) & 
  \end{tabular}} &  & \ \\
\dwln \dwupl & & \dwupr arithmetic\\
  \tmverbatim{}\tmverbatim{}\tmverbatim{\begin{tabular}{ll}
    5 * ( & 5\\
    \ \ \ \tmtextit{if} & \tmtextit{\tmtextbf{True} else}\\
    & \tmtextit{5 * power(5, 1-1)}\\
    \ \ \ ) & 
  \end{tabular}} &  & \ \\
\dwln \dwupl & & \dwupr if-True \\
  \tmverbatim{}\tmverbatim{}\tmverbatim{\begin{tabular}{ll}
    5 * ( & \tmtextbf{5} \ )
  \end{tabular}} &  & \ \\
\dwln \dwupl & & \dwupr arithmetic\\
  \tmverbatim{}\tmverbatim{}\tmtextbf{\tmverbatim{\begin{tabular}{l}
    25
  \end{tabular}}} &  & \ 
\end{tabular}

\

\caption{Executing a recursive function.}
\label{fig:power-exy}
\end{figure*}

Note that Figure~\ref{fig:power-exy} is just one of many possible sequences;
as in our \pyi{2*a+2*b} example,
a number of orders are possible,
and all lead to the same final answer.

\subsection{Discretionary Additional Details}
\label{sec:subst-details}
At the instructor's discretion,
the course can also ``dig deeper'' into details of substitution,
including the internals of the name-to-body rule,
the relationship between name-to-def and name-to-body,
and our definition of ``safe'' expressions.

\subsubsection{Function definition and substitution}
\label{sec:subst-details-funcdef}
In Python, as with many other languages,
functions can be defined with the same notation used for other variables.
For example, Figure~\ref{fig:punct1f-code}'s definition is expressed
in the ``name:type = value'' notation as

\pyi{      recPunct:Callable[[str], str] = (}$\lambda$\pyi{ s: s + ('?' if s[0:4]=='What' else '.')}

We recommend some slightly-non standard prounciation for
instructors who choose to show this to beginners.
When $\lambda$ is used in this way,
we pronounce $\lambda$ ``a function from'',
and \pyi{:} ``to'',
so $\lambda$~\pyi{x}~\pyi{:}~\pyi{x+1} is read
``a function from $x$~to~$x+1$'',
and $\lambda$~\pyi{s}~\pyi{:}~\pyi{s+'?'}
is read
``a function from $s$~to~$s+$\pyi{'?'}''.

By viewing functions in this way,
we can isolate one element of a function call,
by using the name-to-def rule
to turn \eg~\pyi{recPunct('Cool')}
into

\pyi{      (}$\lambda$\pyi{ s: s + ('?' if s[0:4]=='What' else '.')('Cool')}

Note that this requires that the function be expressed as a single expression;
we could recreate the name-to-body's technique of
automatically converting \pyi{if} statements with
and eliminating local variables name-to-def on any local variables,
but it turns out that we can also use $\lambda$ to
introduce local variables in Python.

\subsubsection{Local variable definitions as ``let'' expressions}
\label{sec:subst-details-localdef}
Local variable definitions can also be rewritten as \textit{expressions} by using $\lambda$,
and in this case we can keep the initializer visually near the name by using Python's default-argument mechanism.
So, ``\pyi{a=14;}~$...$\pyi{3*a}$...$'' is equivalent to
\pyi{(}$\lambda$~\pyi{a=14:}~$...$~\pyi{3*a}~$...$\pyi{)()},
assuming that the local variable is used according to the rules of Figure~\ref{fig:decs-conds}.
In this case, $\lambda$ is pronounced ``let'', \pyi{=} is pronounced ``be'' (or even ``be a name for''),
and~\pyi{:} is pronounced ``in'',
so the the above is pronounced ``let $a$~be~14 in $...~3*a~...$''.

Thus, we can retain the original structure of
a function containing local variables,
while still converting the body into a single returned expression.
Similarly, conversion of \pyi{if} statements into expression form
preserves all three components of an \pyi{if} statement,
albeit in a different order.

We currently envision supporting this structure-preserving conversion to function calls,
to let students explore execution in greater detail.
Students can continue to write code with idiomatic Python \pyi{if} statements
and local variables
(within the rules of Figure~\ref{fig:decs-conds}).
They can then use the 3-2-1 SNM to execute a call to an untrusted function via either
the name-to-body rule,
or by applying the name-to-def to a function name,
at which point \orca will automatically convert local variables and \pyi{if} statements into expression form,
to ensure the resulting code is legal Python.
The local variables and function parameters,
can then be handled with the name-to-def rule.
So, when name-to-def is applied to the \pyi{s} parameter in the \pyi{recPunct('Cool')} example
of Section~\ref{sec:subst-details-funcdef},
the value \pyi{'Cool'} replaces use of \pyi{s},
and \pyi{s} and its associated initializer are removed from the parameter and argument lists,
producing

\pyi{      (}$\lambda$\pyi{ : 'Cool' + ('?' if 'Cool'[0:4]=='What' else '.')()}

A separate func-to-body rule is then used to remove zero-parameter $\lambda$
and the associated \pyi{:} and \pyi{( )}, producing

\pyi{      ('Cool' + ('?' if 'Cool'[0:4]=='What' else '.')}

\noindent
and, using other rules, eventually \pyi{'Cool.'}.
Readers familiar with lambda calculus may note that application of name-to-def for each parameter,
followed by func-to-body,
correspond to $\beta$-reduction.

The use of these rules for a recursive function,
or for calls among functions that have identically-named local variables,
can introduce violations of Figure~\ref{fig:decs-conds}'s prohibition against having two variables of the same name in a single scope.
This problem can be avoided by using only carefully-chosen non-recursive examples
or by sticking to the name-to-body rule for function calls;
alternatively, the underlying issue can be addressed with an explicit discussion of scoping,
if the instructor wishes to address this in CS1.

\subsubsection{Scope of Variable Names}
Scoping rules can be quite confusing for beginning programmers,
but deferring this topic could potentially leave students with a partially-formed or flawed understanding of the meaning of program variables.
Some work has been done on possible pedagogy
for faculty who wish to engage the issue of scoping in an introductory course~\cite{fkt:sigcse:2017sma}.

Our basic name-to-body rule of Section~\ref{sec:subst-body}
steers clear of this issue,
by automatically removing all local variable and parameter names.
While correct manual application of the rule could be challenging,
we believe the issue of scoping could be avoided entirely
in a course that uses \orca to apply this rule to carefully-selected examples.

We believe an \orca implementation of the 3-2-1 SNM
could also be used to \textit{illustrate} this issue
while still preventing mistakes.
Specifically, application of the name-to-def rule for a function,
\eg, \pyi{recPunct} above or our recursive \pyi{power} function,
could rename variables as during the same automatic restructuring that converts to expression form,
retaining compliance with the rules of Figure~\ref{fig:decs-conds}.
(In lambda calculus terms, $\alpha$-conversion can be used to ensure that subsequent $\beta$-reduction is legal).
We have previously used subscripts to create fresh names while retaining recognizability~\cite{Wonnacott:fvte17}[Figure~3.8],
and this technique can be used to check student understanding by asking them students to
annotate images from which the subscripts have been omitted.

Alternatively, name-to-def rule could produce code that violates Figure~\ref{fig:decs-conds},
and subsequent applications of name-to-def would then have to limit substitutions to non-shadowed uses of the variable.
As above, this would likely be unworkable with manual application of rules;
\orca could ensure that the rules are followed correctly during automatic execution of the 3-2-1 SNM,
though the rules of scoping could well remain mysterious to students
unless this were paired with manual exploration of the name-to-def rule for small examples.

Note that the issue of variable scope also arises in other contexts,
such as discussions of execution
with the dictionary notional machine.
Significant class time, and space in the course notes,
is spent illustrating how to list parameter and local variable values
\textit{with each entry in a call tree},
to keep track of the separate values for each instance of a variable during recursive calls
(or other calls involving functions with identically-named parameters).
As with a substitution notional machine,
this issue can be largely avoided by careful selection of examples,
if the instructor hopes to postpone this issue until after the first semester.

Thus, the problem of scoping is a fundamental issue of understanding computation,
not a consequence of the use of a \textit{substitution}-based notional machine.

\subsubsection{Issues of termination, strictness, trust, and safety}
The references to ``safe'' code in the rules for the 3-2-1 SNM
help us to capture some ways in which computation differs from algebra.
A program that computes \pyi{x} and then returns \pyi{(x+1)-x} is \textit{not} equivalent to just 1
if the definition of \pyi{x} might run forever or throw an exception.

These issues can be largely ignored,
if we rely on \orca to correctly apply the substitution rules.
Alternatively,
the instructor can emphasize this point,
and present programming idioms that rely on the fact that Python's \pyi{and} operation is not commutative,
as discussed in Section~\ref{sec:subst-vs-algebra}.

\subsection{Comparison to Substitution in Beginning Algebra: Trust and Safety}
\label{sec:subst-vs-algebra}

The 3-2-1 SNM,
and, in particular, the two-column traces we present,
resemble proofs from introductory algebra.
However, our rules differ from those of algebra in their
emphasis on distinguishing \newterm{trusted} and \newterm{untrusted} functions,
and \newterm{safe} and \newterm{unsafe} expressions.
These considerations are vitally important when programming,
regardless of the use of formal verification,
and they arise in several elements of the 3-2-1 curriculum.

These issues can described informally to the students, as was done above,
and their enforcement left to \orca.
Alternatively, they may be presented in some detail.
A function or operation is \newterm{trusted} iff
any call involving appropriate arguments (\ie, meeting its precondition)
will terminate (rather than running forever)
without throwing an exception (as would \pyi{sqrt(-1)},
and with a correct result (\ie, one meeting its postcondition).
Initially, we trust the Python operators and library functions;
trust of user-defined functions is taken up in Section~\ref{sec:symbolic-exec}.

An expression is \newterm{safe} iff all variables it uses are safe
and all functions (and operations) it uses are trusted,
with each called only with argument expressions
that are themselves safe and must meet the function/operation's preconditions.
The evidence of the precondition's safety is normally found in an \pyi{if} expression
containing, or contained with, a potentially-unsafe expression.
So,
the expression
\pyi{2*sqrt( (a-b) )} is not, in isolation, safe,
but the expressions
\pyi{2*(sqrt(a-b) if a>b else sqrt(b-a))} and
\pyi{2*sqrt( (a-b) if a>b else (b-a) )} are safe.
Safety and trust are related, but not identical,
in that the former depends on context ---
I trust my brother, but I would not feel safe in a plane he was landing,
since he doesn't know how to land as plane.

%
%
%

The rules of the 3-2-1 SNM follow from the \newterm{strictness} of elements of the Python language,
and are designed to keep us from transforming
expressions that could fail, or run forever, into a seemingly-correct answer.
For example, we can't turn \pyi{(0 * x)} into \pyi{0} and remove a definition of \pyi{x}
unless the defining expression is safe,
since preceding that calculation with \pyi{x=1//0},
or with \pyi{x=power(2,-3)} and the definition of \pyi{power} from Figure~\ref{fig:power-code},
would keep the program from ever reaching the \pyi{(0 * x)} step.
Our rules are somewhat stricter than is necessary for this purpose,
but our curriculum does not require that we explore this issue.

The asymmetric strictness rules of the Python language
cause our substitution rules to differ somewhat from traditional Boolean algebra,
for example logical conjunction (\ie, $\land$) is usually considered commutative in mathematics courses.
However, in code, one may idiomatically and safely write \pyi{if y>0 and x/y > z: ...},
understanding that this is fundamentally different from the potentially-incorrect \pyi{if x/y > z and y>0: ...}.

This emphasis on considering the validity of each definition and expression
could potentially give students a deeper understanding of mathematics.
For example,
even students who succeeded in high-school algebra can be perplexed
when confronting the classic joke ``algebra proof''
that begins $let~x=y$ and concludes $1+1=1$.
Though we have not studied this effect carefully,
our interactions with students suggest that
those who have been trained to consider the this safety issue
often resolve the ``$1+1=1$'' conundrum quickly once we suggest they consider safety.

\subsection{Comparison to Language and Substitution in HtDP}
\label{sec:subst-vs-BSL}
The curriculum most similar to 3-2-1 is HtDP,
though these both have similarities
to the classic ``Structure and Interpretation of Computer Progarms''\cite{sicp}.
The language and SNM of 3-2-1 differ from those of HtDP in a number of important ways,
however.

3-2-1 uses a subset of Python (or C++ or Java, in some semesters),
while HtDP uses a subset of Racket (Scheme, previously).
Scheme and Racket provide a better match for our pedagogies,
e.g. without distinct statement and expression forms of \pyi{if},
whereas other languages may provide a better match for student enthusiasm.

\unused{  
    Type declarations create notational differences from mathematics,
    and increase the set of concepts that must be presented in the very first code example.
    We have also found that they reduce debugging time and frustration,
    and play an even more fundamental role as student learn to create classes (in our CS2).
    Specifically, when we moved from C++ to Python 2 in both semesters,
    there was a great increase in student confusion between an \textit{object}
    and its \textit{representation},
    for example for a function \pyi{f} with a parameter of type \pyi{Point},
    a student might attempt a call \pyi{f(3, 5)} rather than \pyi{f(Point(3, 5))};
    in C++, the compiler's error message \textit{on the line of the call}
    brought the student's attention to the problem,
    and reinforced the concept that object and representation are different.
    Python 2 would provide a run-time error at the point of the call,
    but the student could eliminate \textit{that} error message by writing \pyi{f((3, 5))},
    and (in our experience) often remained confused about concepts
    while struggling to debug code in which the ``Point'' \pyi{(3, 5)}
    refused to respond to messages for class \pyi{Point}.
    Possibly the HtDP pedagogy steers students clear of this pitfall,
    but we have simply used \code{mypy} and/or Java to reassert the role of language tools.
}

The 3-2-1 SNM differs from the HtDP notional machine in one fundamental way:
the former does not enforce a specific substitution order,
whereas the latter evaluates all arguments, from left to right, before a function is called
(\ie, it uses left-to-right applicative-order evaluation).
The freedom to explore different execution orders
creates both opportunities and perils for students in the 3-2-1 curriculum.
Most significantly,
it supports the symbolic evaluation is fundamental to
our path from execution, to deduction of abstract properties,
and on to proofs with mathematical terminology and structure (see Section~\ref{sec:symbolic-exec}).

Exploration of different execution orders can potentially be either enlightening or confusing.
Students may be concerned that they are accidentally picking a ``wrong'' order,
or notice that they can create an unending chain of substitution by,
\eg, applying the name-to-body rule to the \pyi{power} function
of Figure~\ref{fig:power-code} without ever choosing another rule.
We believe (without \textit{yet} having made a formal study)
that this confusion can be addressed via instruction and/or tool support:
Students could start using \orca with an ``auto-step'' button
corresponding to left-to-right applicative order
(aligning this initial use of the 3-2-1 SNM with HtDP's BSL)
and then branch out with a different auto-step
(\eg, right-to-left applicative order, or even normal-order evaluation),
or jump to manual rule selection.
Alternatively, the instructor could show students how to achieve these orders as they choose steps.
The freedom to choose in different ways can let students explore issues of efficiency,
\eg, exploring the impacts of forward substitution or common subexpression factoring,
and thereby gaining a richer understanding of what can actually happens
when we take a source-code program and run it,
preparing them for discussions of memory models in later courses on concurrent programming,
or of optimization in high-performance computing,
as well as supporting theoretical foundations of computing.
Assuming, that is, they don't simply become confused.
We feel this topic deserves field study.

The 3-2-1 SNM's name-to-spec and name-to-body rules also
emphasizes the distinction between trusted and untrusted code,
while also showing the connection between student-written code and library functions
(which, are, of course, just functions someone wrote, but which we now trust).

After some use of this rule,
students may wonder ``can I trust my own function'',
which leads us directly to the deduction of abstract properties of code,
and the of mathematical proof techniques.

\section{Transitioning to Proofs}
\label{sec:symbolic-exec}
If we use the flexibility of the 3-2-1 SNM
to delay the substitution of some variable definitions,
we can use it to perform \textit{symbolic} execution of code.
For example, we could wait to substitute
the definition of \pyi{x} in Figure~\ref{fig:power-exy}
until after the use of the if-True rule.
In this case,
the if-True rule would produce \pyi{x * (x)} rather than \pyi{5 * (5)},
and after substituting \pyi{(3+2)} for \pyi{x} we would once again arrive at \pyi{25}..

\begin{figure*}[tb]
  \centering
  
\begin{tabular}{l|lp{7.0cm}}
  \tmverbatim{\begin{tabular}{l}
    power(x, \tmtextit{y})
  \end{tabular}} & \quad & \ \\
\dwln \dwupl & & \dwupr name-to-def: \tmverbatim{y}\\
  \tmverbatim{\begin{tabular}{l}
    \tmtextit{power(x, \tmtextbf{2})}
  \end{tabular}} &  & \ \\
\dwln \dwupl & & \dwupr name-to-body: \tmverbatim{power(x, 2)}\\
  \tmverbatim{\begin{tabular}{ll}
    ( & x\\
    if & \tmtextit{2 == 1} else\\
    & x * power(x, \tmtextit{2-1})\\
    ) & 
  \end{tabular}} &  & \ \\
\dwln \dwupl & & \dwupr arithmetic: $2 = 1$ and $2 - 1$\\
  \tmverbatim{}\tmverbatim{}\tmverbatim{}\tmverbatim{\begin{tabular}{ll}
    ( & \tmtextit{x}\\
    \tmtextit{if} & \tmtextit{\tmtextbf{False} else}\\
    & x * power(x, \tmtextbf{1})\\
    ) & 
  \end{tabular}} &  & \ \\
\dwln \dwupl & & \dwupr if-False\\
  \tmverbatim{}\tmverbatim{}\tmverbatim{}\tmverbatim{}\tmverbatim{}\tmverbatim{}\tmverbatim{\begin{tabular}{l}
    x * \tmtextit{power(x, 1)}
  \end{tabular}} &  & \ \\
\dwln \dwupl & & \dwupr name-to-body: \tmverbatim{power(x, 1)}\\
  \tmverbatim{}\tmverbatim{}\tmverbatim{\begin{tabular}{ll}
    x * ( & \tmtextbf{x}\\
    \ \ \ \tmtextbf{if} & \tmtextbf{\tmtextit{1 == 1} else}\\
    & \tmtextbf{x * power(x, 1-1)}\\
    \ \ \ ) & 
  \end{tabular}} &  & \ \\
\dwln \dwupl & & \dwupr arithmetic: $1 = 1$\\
  \tmverbatim{}\tmverbatim{}\tmverbatim{\begin{tabular}{ll}
    x * ( & x\\
    \ \ \ \tmtextit{if} & \tmtextit{\tmtextbf{True} else}\\
    & \tmtextit{x * power(x, 1-1)}\\
    \ \ \ ) & 
  \end{tabular}} &  & \ \\
\dwln \dwupl & & \dwupr if-True\\
  \tmverbatim{}\tmverbatim{}\tmverbatim{\begin{tabular}{ll}
    x * ( & \tmtextbf{x} \ )
  \end{tabular}} &  & \ \\
\dwln \dwupl & & \dwupr algebra: $x \cdot x = x^2$\\
  \tmverbatim{}\tmverbatim{}\tmtextbf{\tmverbatim{\begin{tabular}{l}
    x ** 2
  \end{tabular}}} &  & \ 
\end{tabular}

\

\caption{Executing \pyi{power(x, 2)} without a value for \pyi{x}.}
\label{fig:power-pf0}
\end{figure*}
However, if we stop before putting in the value of \pyi{x},
as is actually done in Figure~\ref{fig:power-pf0},
we learn something much more interesting:
when \pyi{power}'s second parameter is 2,
\pyi{power(x,y)} will produce $x^2$ for any safe expression $x$.
Thus, with a simple variation of the ways we use substitution rules we've already seen,
we have made a fundamental shift in our ability explore the meaning of code:
we are now reasoning about abstract properties, rather than execution for specific values.

\begin{figure*}[tb]
  \centering
  
\begin{tabular}{l|lp{3.0cm}}
  \tmverbatim{\begin{tabular}{ll}
    $(\nobracket$ & \tmtextit{power(x, y)}\\
    if & y > 0 else\\
    & ERROR\\
    ) & 
  \end{tabular}} & \quad & \ \\
\dwln \dwupl & & \dwupr name-to-body\\
  \tmverbatim{\begin{tabular}{ll}
    $(\nobracket$ & \begin{tabular}{ll}
      ( & x\\
      if & y == 1 else\\
      & x * \tmtextit{power(x, y-1)}\\
      ) & 
    \end{tabular}\\
    if & y > 0 else {\textdots}
  \end{tabular}} &  & \
  
  \ \\
\dwln \dwupl & & \dwupr name-to-spec-simpler\\
  \tmverbatim{\begin{tabular}{ll}
    $(\nobracket$ & \begin{tabular}{ll}
      ( & x\\
      if & y == 1 else\\
      & \begin{tabular}{ll}
        x * \tmtextbf{(} & x ** (y-1)\\
        \ \ \ \tmtextbf{if} & \tmtextit{y-1>0 and y>1 and y>y-1} else\\
        & ERROR\\
        \ \ \ \tmtextbf{)} & 
      \end{tabular}\\
      ) & 
    \end{tabular}\\
    if & y > 0 else {\textdots}
  \end{tabular}} &  & \ \\
\dwln \dwupl & & \dwupr algebra (consider-tests)\\
  \tmverbatim{}\tmverbatim{}\tmverbatim{\begin{tabular}{ll}
    $(\nobracket$ & \begin{tabular}{ll}
      ( & x\\
      if & y == 1 else\\
      & \begin{tabular}{ll}
        x * ( & x ** (y-1)\\
        \ \ \ if & \tmtextbf{True} else\\
        & ERROR\\
        \ \ \ ) & 
      \end{tabular}\\
      ) & 
    \end{tabular}\\
    if & y > 0 else {\textdots}
  \end{tabular}} &  & \ \\
\dwln \dwupl & & \dwupr if-True\\
  \tmverbatim{}\tmverbatim{}\tmverbatim{\begin{tabular}{ll}
    $(\nobracket$ & \begin{tabular}{ll}
      ( & x\\
      if & y == 1 else\\
      & x * ( \tmtextit{\tmtextbf{x ** (y-1)}} )\\
      ) & 
    \end{tabular}\\
    if & y > 0 else {\textdots}
  \end{tabular}} &  & \ \\
\dwln \dwupl & & \dwupr algebra\\
  \tmverbatim{}\tmverbatim{\begin{tabular}{ll}
    $(\nobracket$ & \begin{tabular}{ll}
      ( & \tmtextit{x}\\
      if & y == 1 else\\
      & \tmtextbf{x ** y} )\\
      ) & 
    \end{tabular}\\
    if & y > 0 else {\textdots}
  \end{tabular}} &  & \ \\
\dwln \dwupl & & \dwupr algebra \\
  \tmverbatim{\begin{tabular}{ll}
    $(\nobracket$ & \begin{tabular}{ll}
      ( & \tmtextbf{x ** y}\\
      if & y == 1 else\\
      & x ** y )\\
      ) & 
    \end{tabular}\\
    if & y > 0 else {\textdots}
  \end{tabular}} &  & \ \\
\dwln \dwupl & & \dwupr if-irrelevant\\
  \tmverbatim{\begin{tabular}{ll}
    $(\nobracket$ & ( \tmtextbf{x**y} )\\
    if & y > 0 else\\
    & ERROR\\
    ) & 
  \end{tabular}} &  & \ \\
\dwln \dwupl & & \dwupr \tmtextit{Q.E.D.}
\end{tabular}

\

\caption{Verifying correctness of \pyi{power}.}
\label{fig:power-pf1}
\end{figure*}
It seems natural to wonder whether we might be able 
to show that, for any value of \pyi{y},
\pyi{power(x, y)} will be $x^y$.
Unfortunately, that claim is not true \ldots
\pyi{power} does not produce this result if its second argument is,
for example, \pyi{-5}.
However, we can show that \pyi{power(x, y)} will be $x^y$
\textit{whenever \pyi{y} is a legal value},
\ie, a positive integer,
and \pyi{x} and \pyi{y} are safe
(\ie, in any application of the name-to-spec rule to \pyi{power}).
Thus, we will begin our exploration of the correctness of \pyi{power}
with an arbitrary call \pyi{power(x, y)}
that is placed within the context of this precondition check,
and work toward the form that would \textit{begin} the name-to-spec rule,
as shown on the first and last lines of Figure~\ref{fig:power-pf1}.

However, the tools we've seen so far won't let us connect the starting and ending lines of Figure~\ref{fig:power-pf1}.
We can apply the name-to-body rule once, as shown,
but additional use of name-to-body doesn't help after this point:
we'll never get to the \pyi{1==1} expression that unlocked our use of if-True in Figure~\ref{fig:power-pf0},
and thereby ended the substitution process.
Since name-to-body does not help, we might hope to try the name-to-spec rule.
We could provide definitions of \textit{expected result},
\ie, \pyi{power(b, e)} should return \pyi{b**e},
and \textit{argList is legal},
\ie, \pyi{e>0}.
However, we are not allowed to apply the rule,
since \pyi{power} is not (yet) a trusted function.

To complete the sequence of substitutions from \pyi{power(x, y)} to \pyi{x**y},
we'll need to introduce our third (and final) function-call rule.
This rule lets a recursive function trust \emph{itself} during this verification step,
but only when used for a \textit{simpler instance of the problem} that is solved by the function.
We treat the definition of the term ``simpler'' as many other have done before,
by introducing the idea of a \textit{progress expression}:
an expression that decreases in every recursive call,
and for which no recursive calls will be made when it is below some minimum value.
For the recursive algorithms presented in the 3-2-1 curriculum,
the progress expressions are typically extremely simple,
\eg, the length of a string parameter, or the value of an integer parameter
(such as \pyi{e} in our \pyi{power} function).
Later, this concept is generalized to include the de-construction of recursive data types,
when structural recursion is introduced for recursive data types in later in the 3-2-1 curriculum.
Our new rule is

\begin{samepage}
\begin{description}
    \item[name-to-spec-simpler]  An expression of the form \pyi{func(argList)},
                      when found within the body of \pyi{func} itself
                      during the check to see if we can trust \pyi{func},
                      can be replaced with\\
                      \pyi{  (} \textit{expected result}\\
                      \pyi{  if} \textit{argList is legal} \pyi{and}
                                  \textit{progr > pmin} \pyi{and} \textit{prog > progr'}\\
                      \pyi{  else ERROR )}\\
                      where \textit{progr} is the progress expression,
                      \textit{pmin} is a minimum value at or below which there will be no recursive call, and 
                      \textit{progr'} is the progress expression with the parameters for the recursive call.
\end{description}
\end{samepage}

As was the case with name-to-body and name-to-spec,
the we only apply the name-to-spec-simpler rule if
the expressions used to define local variables
and the expressions for \pyi{argList}, are safe
(assuming the variables introduced at the start of the process,
\eg, \pyi{x} and \pyi{y} of Figure~\ref{fig:power-pf1},
are safe).
\unused{
    As was the case for our latest name-to-body rule (of Section~\ref{sec:subst-body}),
    the actual application of the name-to-spec-simpler rule might prove overly-complex in the context of a hand-written sequence,
    but the conceptual foundation is clear:
    A regular call requires valid parameters; a recursive call requires valid parameters and progress toward ending the recursion.
    And, as noted, our latest set of rules is designed to improve our curriculum
    in the context of \orca.
}

Figure~\ref{fig:power-pf1} also raises interesting pedagogical questions about the degree of desirable automation,
when it comes to tracking what is known in a specific program context.
Note that the third simplification step, ``algebra'', simplifies \pyi{(y-1>0 and y>1 and y>y-1)} to \pyi{True}.
While the last part of that conjunction is true for any integer $y$,
the first two follow from the facts that $y>0 \land y \neq 1$, and the code contains is no definition of \pyi{y}.
The fact that $y>1$ could be performed automatically within \orca
(as per the ``context lookup'' example of the upcoming Figure~\ref{fig:orca-example-action}),
allowing students to perform the simplifications of Figure~\ref{fig:power-pf1}.
However, this frees students from some of the deductive steps that we'd like them to master.
Thus, we could choose to require another substitution rule to force students to identify the relevant conditions:

\begin{samepage}
\begin{description}
    \item[consider-tests]  Within an expression of the form  $(X$ \pyi{if} $C$ \pyi{else} $Y)$,
                           simplify $X$ based on the truth of $C$, or $Y$ based on $\neg C$.
                           Specifically, a boolean expression within $X$ or $Y$ can be replaced
                           with \pyi{True} or \pyi{False} if it is ensured (or contradicted)
                           by $C$ or a combination of conditions in nested \pyi{if}'s,
                           or a variable used within $X$ or $Y$ can be replaced with a single
                           value of the conditions restrict it to a single value.
\end{description}
\end{samepage}

For example, in Figure~\ref{fig:power-pf1},
the consider-tests rule, when applied to
the conditionals \pyi{y>0} (the outer \pyi{if} being true)
and \pyi{not y==1} (the inner \pyi{if} being false),
replaces \pyi{y-1>0} and the equivalent \pyi{y>1} with \pyi{True}.
Application of consider-tests to \pyi{y==1}
could replace occurrences of \pyi{y} with \pyi{1} in the true branch of \pyi{if y==1},
though in this code there are no such occurrences.

We can now complete the exploration of \pyi{power(x, y)},
by making use of the name-to-spec-simpler rule,
appropriate pre- and postconditions (from above; expressed as a specific ``expected result''),
and definitions of \textit{progr} and \textit{pmin},
\ie, \pyi{e} and \pyi{1},
and (optionally) a consider-tests rule.
This complete sequence is shown in Figure~\ref{fig:power-pf1}.
Thus, we have shown that
any use of \pyi{power(x, y)} that occurs where \pyi{y>0} can be transformed into \pyi{x**y}.
Once this has been established,
we say that we have \textit{verified} the correctness of the
\pyi{power} function, and we can add it to our set of trusted functions
and apply the name-to-spec rule to it.

In previous semesters, progress expressions and termination were handled as
ancillary steps in our system for formal verification,
which was forced to adopt a goal/argument structure like that of mathematical proof,
and was thus somewhat divorced from the substitution notional machine.
Once again, we see automated application of our rules
as the key to successful use of rules that capture the concepts at the appropriate level of abstraction.

The discussion above covers our core method of
introducing formal reasoning for pure-functional code
in the context of program execution.
We hope that every student in our introductory class will
be able to perform verification in this manner,
possibly by using \orca,
by the end of the first semester.

The instructor may or may not choose to introduce the term \textit{proof}
at this point,
depending on the class' comfort with mathematics,
and the instructor's feelings about applying a term usually reserved for the purely abstract
entities of mathematics to a real-world object, \ie, a Python function in the computer's memory.
Even if we verify our functions,
they can potentially fail due to problems in the language system, operating system, or hardware,
or even a user accidentally hitting the power button.
Regardless of the terminology chosen,
the students will have been given significant exposure to the process of deductive reasoning,
entirely within the structure of understanding code and computation.

\subsection{Discretionary Subtleties}
\label{sec:symbolic-additional}

As noted above,
Figure~\ref{fig:power-pf1} might be considered too simple for our purposes,
if we want students to understand the process of formal deduction.
In principle,
we could support exactly the sequence of Figure~\ref{fig:power-pf1},
or require something like the consider-context rule mentioned above,
or let students start without consider-context,
and then require it later.
Possibly, different approaches are best for different student populations.

Those familiar with the teaching of mathematics may already have identified
Figure~\ref{fig:power-pf1} as an inductive proof,
but our goal is to get any ``math-phobic'' students through this material
without forcing them to associate it with any prior negative experience
they may have had with mathematics or proofs
(in our experience, students sometimes seek out CS
because they are looking for a science that does not require mathematics).
For students who are comfortable with
the traditional conjecture/proof terminology and structure of mathematics,
including inductive proof,
the instructor can of course discretely make this connection earlier.

Whether or not students have succeeded with the details of induction before,
they may (appropriately!) become concerned about circular reasoning.
In such cases,
it may help to point out that it would be incorrect to define,
\eg, \pyi{power(5,2)} in terms of \pyi{power(5,2)},
but there is no problem defining \pyi{power(5,2)} in terms of \pyi{power(5, 1)}
(as long as \pyi{power(5, 1)} is not, itself, circularly defined).
This can lead to a general discussion of the role progress expressions and base cases for recursion,
in ruling out just this sort of situation.


If the instructor chooses to introduce the conjecture/proof structure,
and relate it to code verification such as in Figure~\ref{fig:power-pf1},
slight variations on our examples can be used to motivate mathematical sophistication.
In some cases,
it may be much easier to describe the \textit{properties} of a function's result,
than to specify a unique correct answer.
For example,
we might wish to say that a function returns a minimum-length string from a collection of strings,
without requiring that a particular value be produced if several strings have that minimum length.
Or, we might wish to describe the result of \pyi{sqrt(x)} as
``a non-negative value with the property that \pyi{sqrt(x)*sqrt(x)} gives back \pyi{x}
(within the limits of the real-number approximations)''.

Such cases are best addressed in the traditional mathematical structure of
``given these axioms and definitions, establish this result'',
rather than the ``substitute for the function call until you reach the expected result'' framework
of the 3-2-1 SNM rules presented here.
However, if/when the instructor wishes to make this change,
the process of formal deduction (via substitution) will already be familiar to the students.
Note that prior offerings of the 3-2-1 curriculum
used a different function-call rule,
and the verification system did not fit cleanly within the substitution framework.
We are optimistic that the rules given here will provide all students with
a clear connection between program execution and abstract reasoning about programs,
and that this abstract reasoning can then be translated into a proof structure
at the appropriate time, either CS1 or discrete math, as dictated by institutional and student needs.

Verification of imperative functions,
like those involving code with non-constructive postconditions,
can be expressed more naturally in the conjecture/proof framework,
should the instructor wish to reinforce the software engineering and verification elements
that were introduced in the functional programming exercises of the 3-2-1 curriculum.
This framework also extends naturally to data structures;
Liskov and Guttag provide an excellent treatment of this material in the context of Java
(or Clu) programming~\cite{aspd,aspdj}.

Instructors who wish to demonstrate code verification within this structure can switch to an appropriate
more-general logic within \orca for later work in (or after) CS1,
as discussed in Section~\ref{sec:discrete-mathematics}.

We suspect that many of the concepts of Sections~\ref{sec:substitution} and \ref{sec:symbolic-exec}
could be reformulated in terms of symbolic execution with a dictionary notional machine.
However, we have not explored the possibility of doing so,
as this framework is further from the usual framework of mathematical proof.

\section{Symbolic Reasoning in Discrete Mathematics}
\label{sec:discrete-mathematics}
In contrast to Haverford's CS1 course, Grinnell does not introduce code reasoning in its introductory course.
Instead, we opt to introduce these elements in its discrete mathematics course, building on the functional programming core that it develops in CS1.
By introducing code reasoning into discrete mathematics course, we hope to address two issues we have seen with our students:
\begin{enumerate}
  \item Students don't understand how to construct a well-formed mathematical proof in later theoretical courses.
  \item Students view mathematical proof as an activity completed unrelated to their usual business of computer programming.
\end{enumerate}

In \emph{How to Prove It}, Velleman likens the structure of mathematical proof to computer programs~\cite{velleman_2006}.
A computer program is composed of smaller program elements according to a set of rules which govern which combinations of elements are valid.
Likewise, a proof is composed of smaller proof elements according to a set of rules
It is these rules which ought to be explicated in detail rather than acquired indirectly through mere exposure to examples.
Velleman and others follow this pedagogy of \emph{explicit proof education} by first introducing mathematical logic as the rule set for primitive proof formation and then explaining how these rules translate to real-world proofs.

Our pedagogy at Grinnell follows suit, first introducing logic as the basis for proof.
However, it differs from these other approaches in that it uses code reasoning as its domain for introducing proof rather than more traditional domains such as calculus or number theory.
For example, here is the definition of a function in OCaml, the functional language that we use in the course:
\begin{ocamlcode}
let andb (b1:bool) (b2:bool) =
  match b1 with
  | true  -> b2
  | false -> false
\end{ocamlcode}
a simple proposition we assert over that function
\begin{claim}
  $\forall \oi{b}.\,\oi{andb b false} = \oi{false}$,
\end{claim}
that postulates that \oi{andb} always produces \oi{false} when its second argument is \oi{false}, and a proof of that proposition:
\begin{proof}
Because \oi{b} is a boolean, it can be either be \oi{true} or \oi{false}.
Suppose it is \oi{true}.
Then we have:
\begin{align*}
& \oi{andb true false} \\
\longrightarrow\;& \oi{match true with true -> false | false -> false} \\
\longrightarrow\;& \oi{false}
\end{align*}
And suppose it is \oi{false}.
Then we have:
\begin{align*}
& \oi{andb false false} \\
\longrightarrow\;& \oi{match false with true -> false | false -> false} \\
\longrightarrow\;& \oi{false}
\end{align*}
\end{proof}
where the steps of evaluation are justified by the standard small-step evaluation semantics we give for OCaml.

The core of this proof is simple substitutive evaluation, as in Section~\ref{sec:substitution},
but it can still be used to motivate formal logical reasoning, \eg, in this case by proving a universally quantified claim by assuming an arbitrary value \oi{b} and then performing case analysis on the possible values of \oi{b}.

\subsection{Inductive Reasoning}
\label{sec:dm-induction}

One kind of proof that students find especially difficult to master is inductive proof.
In our course, we address this problem by introducing \emph{structural recursion} over lists first, rather than mathematical induction over the natural numbers.
This leverages our students' prior experience with recursive functions from CS1 and gives immediate motivation for inductive reasoning: induction is the primary verification technique for pure, functional programs.

We motivate the need for induction by analyzing recursive functions in the context of the substitutive evaluation model described above.
For example, consider the \oi{append} function over lists:
\begin{ocamlcode}
  let rec append (l1:'a list) (l2:'a list) : 'a list =
    match l1 with
    | []       -> l2
    | x :: l1' -> x :: append l1' l2
\end{ocamlcode}
A property of \oi{append} is that it ought to respect the lengths of the lists.
That is:
\begin{claim}[Append Preserves Length]
  $\forall \oi{l1}, \oi{l2}.\,\oi{length}\;(\oi{append}\;\oi{l1}\;\oi{l2}) = \oi{length l1 + length l2}$.
\end{claim}
In our curriculum, a first attempt at proving this claim involves symbolically evaluating $\oi{length (append l1 l2)}$.
Because \oi{l1} and \oi{l2} are universally quantified, we can't assume particular values for them.
Instead, we have to perform \emph{case analysis} on \oi{l1}---is \oi{l1} empty or non-empty?\footnote{%
  Note that case analysis on \oi{l2} does not get us anywhere because it is not the argument that is
  being pattern matched in the function.
}
\begin{itemize}
  \item If \oi{l1} is empty (\ie, \oi{l1 = []}), so the left-hand side of the equality evaluates to
    \[
      \oi{length (append [] l2)} \longrightarrow^* \oi{length l2}.
    \]
    and \oi{length [] + length l2 = length l2} on the right-hand side.
  \item If \oi{l1} is non-empty, then \oi{l1 = x :: l1'} for some element \oi{x} and (sub-)list \oi{l1'}.
    Evaluation yields:
    \begin{align*}
      & \oi{length (append (x :: l1') l2)} \\
      \longrightarrow^*\;& \oi{length (x :: append l1' l2)} \\
      \longrightarrow^*\;& \oi{1 + length (append l1' l2)}.
    \end{align*}
    The last step of evaluation is justified by the definition of \oi{length}---a list that is a cons of the form \oi{x :: l} has length \oi{1 + length l}.
\end{itemize}
While we are able to prove the empty case, the non-empty case leaves us with the goal:
\[
  \oi{1 + length (append l1' l2) = 1 + length l1' + length l2}.
\]
This is the original goal but with \oi{l1} unrolled by one element.
If we try to evaluate the program further, we simply continue unrolling \oi{l1}, but because \oi{l1} is universally quantified, we don't know when to stop!

We then introduce structural induction as the way to break this infinite chain of reasoning.
Structural induction then becomes \emph{case analysis} coupled with an \emph{inductive hypothesis} that we can invoke on sub-structures in our proof.
We note that this is sound as long as we know that our list is finite in size\footnote{%
  Which is an interesting discussion we bring up in its own right!
  The definition of an OCaml list coupled with the the langauge's strict semantics guarantees that values of the list type are always finite as long as our computations are terminating.
}, and our \oi{append} function only makes recursive calls to structurally smaller lists so that it is guaranteed to terminate.

A complete inductive proof of the claim otherwise proceeds normally:
\begin{proof}
  We prove the claim by induction on \oi{l1}.
  Consider the possible shapes of \oi{l1}:
  \begin{itemize}
    \item \oi{l1 = []}, then $\oi{length (append [] l2)} \longrightarrow^* \oi{length l2}$ and $\oi{length [] + length l2} \longrightarrow^* \oi{length l2}$.
    \item \oi{l1 = x :: l1'} and our inductive hypothesis states that $\oi{length (append l1' l2)} = \oi{length l1'} + \oi{length l2}$.
      Then:
      \begin{align*}
        & \oi{length (append (x :: l1') l2)} \\
        \longrightarrow^* & \oi{length (x :: append l1' l2)}. 
      \end{align*}
      From our inductive hypothesis, we know that \oi{append l1' l2} produces a list of length \oi{length l1' + length l2}.
      Because \oi{l = x :: l1'}, we know that we obtain \oi{l} by consing one more element onto \oi{l1'}, thus we know that:
      \[
        \oi{length (x :: append l1' l2)} = \oi{length l1 + length l2}.
      \]
  \end{itemize}
\end{proof}
Note that this proof is similar to a correctness proof of a recursive function
from Section~\ref{sec:symbolic-exec},
except that it is presented in a more mathematically traditional way (with prose and induction hypothesis rather than pre-/post-conditions).
The connection between the two styles is that our original claim is a post-condition of \oi{append} and the inductive hypothesis is how to apply this post-condition in the special case when we are \emph{currently trying to prove that the post-condition holds}.
The progress condition is implicit in the fact that we can only apply our induction hypothesis on structurally smaller lists which is in turn forced by our case analysis on \oi{l1}.

With these proof techniques firmly established with code reasoning, our discrete mathematics course transitions in its second half to a more traditional course exploring sets, functions, probability, and graphs.
In the future, we would like to identify and apply similar sorts of domains as we have done for logic and code reasoning to make our presentation of these topics as compelling and well-motivated to our students as possible.

\section{Tools to Support Symbolic Reasoning}
\label{sec:orca}
We have discussed our rationale for adopting code reasoning in our introductory programming and discrete math courses.
However, this approach is not without its downsides.
In particular, paper-and-pencil proofs, especially of this sort, impose a significant burden on the students and instructors:
\begin{itemize}
  \item \textbf{Students} must author logically consistent proofs without any feedback or ability to check their work short of submitting their assignment to the instructor.
  \item \textbf{Instructors} are tasked with manually checking all of these proofs and providing quality feedback to the students, a process that is difficult to scale up in the face of computer science's increasing enrollments.
\end{itemize}
The analogy we draw here is to programming: it is inconceivable that we would require our students to write programs on paper and then hand-check their execution after-the-fact without any software support.
In the realm of programming, we have compilers, interpreters, test suites, IDEs, \etc, that provide rich feedback to the student through the entire program development process and help automate the process of grading their work.
What software tools exist to help us manage proofs in the same way that these tools help us manage programs?

There are a multitude of such \emph{proof assistants}, software packages that help users author and check proofs.
Such tools have been used by working mathematicians for large scale developments such as the proof of the four color theorem~\cite{gonthier_2008}, the Feit-Thompson theorem~\cite{gonthier_2013}, the Kepler conjecture~\cite{hales_2017}, and full-stack verification of a compiler~\cite{leroy_2009}.
These tools have also been used in educational contexts such as graduate-level programming languages courses~\cite{pierce_sf1, pierce_sf2} and undergraduate discrete mathematics courses~\cite{greenberg_2019}.

\subsection{Discrete Mathematics with Coq}

In the first semester that we offered this new version of discrete mathematics, students worked purely on pen-and-paper.
In the second semester, the second author integrated the Coq Proof Assistant~\cite{coq-2004} into the code reasoning portion of their course described in~\autoref{sec:discrete-mathematics}.
Coq itself contains a typed, functional programming language, \emph{Gallina}\footnote{%
  Really, a full-spectrum dependently typed language although we only use the subset of the language that corresponds to OCaml: functions and algebraic datatypes with pattern matching.
}, so we translated our code examples to this language and had students prove properties of these programs in Coq directly.

To give a flavor of the differences, here is the \oi{append} proof as realized in Coq:
\begin{coqcode}
Theorem append_length : forall l1 l2,
    length (append l1 l2) = length l1 + length l2.
Proof.
  intros l1 l2.
  induction l1.
  + simpl. reflexivity.
  + simpl. rewrite IHl1. reflexivity.
Qed.
\end{coqcode}
Coq features \emph{proof tactics} which automate the generation of machine checked proofs.
Coq's basic tactics mimic what we would write on paper, \eg,
\begin{itemize}
  \item \coqi{intros l1 l2} introduces the universally quantified variables \coqi{l1} and \coqi{l2} into the proof as unknown values.
  \item \coqi{induction l1} performs induction over \coqi{l1}.
  \item \coqi{simpl} tells Coq to simplify the goal, equivalent to the $\longrightarrow^*$ stepping relation we used in the paper proof.
  \item  \coqi{rewrite IHl1} rewrites the goal according to our induction hypothesis.
  \item \coqi{reflexivity} discharges the proof obligation if the current proof goal is of the form \coqi{e = e} for some expression \coqi{e}.
\end{itemize}
As the student enters their proof, Coq reports the \emph{proof state} at every step.
\autoref{fig:proof-state} gives the proof state that Coq reports at the beginning of the non-empty case for \coqi{l1}.
The proof state contains the set of \emph{assumptions}---\coqi{n}, \coqi{l1}, \coqi{l2}, and \coqi{IHl1}---above the line and the \emph{proof goal} below the line.

\begin{figure}
  \begin{center}
  \begin{minipage}{0.5\textwidth}
  \begin{coqcode}
n : nat
l1, l2 : list
IHl1 : length (append l1 l2) = length l1 + length l2
============================
length (append (Cons n l1) l2) = length (Cons n l1) + length l2
  \end{coqcode}
  \end{minipage}
  \end{center}
  \caption{Example Coq proof state.}
  \label{fig:proof-state}
\end{figure}

From our informal survey of students at the end of the semester, we learned that students found the feedback that Coq provides---the proof state as well as whether the proof is valid---immensely helpful.
The instructor also found their grading burden significantly reduced thanks to Coq's proof checking and regularization of the format of proofs.
These findings are consistent with others experiences in using Coq in the classroom~\cite{pierce_2009, greenberg_2019}.

However, we found some substantial downsides with using Coq in our course:
\begin{itemize}
  \item \textbf{Navigating General-purpose Logics}: to handle broad range of propositions that mathematicians might pose, Coq and other proof assistants are built upon general-purpose logics such as the Calculus of Inductive Constructions~\cite{coquand_1988}.
  In some cases, such as with program correctness, the logic allows for seamless verification.
  But in other cases, the mathematician must spend significant amounts of effort encoding their proof domain into the logic, \eg, encoding combinatorics and probability~\cite{haven_2012}.
  These encodings rapidly become domain-specific languages in their own right, complete with the learning curves of picking up new syntax and mapping one's ideas (in this case, a proof) into a form appropriate for the encoding.
  \item \textbf{Baroque Syntax}: Coq's syntax comes the ML lineage of programming languages.
  While elegant and compact, it is unfamiliar to most students thus inducing a significant learning curve in the classroom.
  \item \textbf{Excessive Automation}: tools such as Coq designed for mathematically-sophisticated users
  typically do not show the excruciating level of detail inherent in single-stepping our substitution notional machine,
  so Figures~\ref{fig:ints-arith}, \ref{fig:ints-arith2}, and \ref{fig:punct1} might end up as a single command (\eg, Coq's \coqi{simpl} tactic which simplifies expressions as much as possible), illustrating nothing for the students.
  Instructors either need to adapt their pedagogy to the particulars of the tool (potentially missing out on learning opportunities such as the single-stepping described above) or wrestle with the tool to try to get the behavior they want from the system.
\end{itemize}

\subsection{Introducing Orca}

To address these concerns, in tandem with our pedagogy, we are also developing \orca, a proof assistant for undergraduate education~\cite{osera-blocks-2017, chen_2017}.
\orca is composed of two components:
\begin{enumerate}
  \item A logic programming-like core for carrying out deductive reasoning and an embedded domain-specific language in Haskell for authoring logics.
  \item A web-based blocked-based language for authoring proofs in \orca.
\end{enumerate}
\orca is still in on-going development to support the courses at Haverford and Grinnell.
Here, we briefly discuss the current design of \orca.

\subsubsection{Core Orca}
The core of \orca is a simple deductive inference engine over function symbols and relations between those symbols:
\[
\begin{tabu}{lcll}
  s & \bnfdef & f(s1, \ldots, sk) & \text{Function symbols} \\
  R & \bnfdef & r(s1, \ldots, sk) & \text{Relations}.
\end{tabu}
\]
Constants such as strings or numbers are represented as nullary symbols, \ie, symbols with no arguments.

\orca manages a collection of user-defined logics, each of which define:
\begin{itemize}
  \item \emph{Sorts}: valid function symbols along with their signatures,
  \item \emph{Relations} along with their signatures,
  \item \emph{Relation actions} which define, for each relation, how to perform deduction with that relation.
\end{itemize}
The most common of these relation actions is the \emph{inference action} which allows the user to perform deductive inference according to a collection of user-defined inference rules.
Relation actions are defined within the embedded domain-specific language, so the user has the full power of the Haskell language to decide how \orca will handle each relation.

\newcommand{\evenr}{\mathsf{even}\xspace}

As a simple example, we might define a simple deductive logic for determining whether a natural number is even.
In Grinnell's discrete mathematics course, this logic is used as a first example for introducing natural deduction-style reasoning.
We define a unary relation (\ie, a predicate) $\evenr(n)$ which proposes that the natural number $n$ is even along with the following inference rules for $\evenr$:
\begin{gather*}
\inferrule[even-zero]
  {}
  {\evenr(0)} 
\qquad
\inferrule[even-nonzero]
  {\evenr(n)}
  {\evenr(n+2)} 
\end{gather*}
\textsc{even-zero} assers that $0$ is even and \textsc{even-nonzero} asserts that $n+2$ is even whenever $n$ is even.
In this logic, a proof that $4$ is even might be written as follows:
\begin{claim*}
  $4$ is even.
\end{claim*}
\begin{proof}
  By \textsc{even-nonzero}, to prove $4$ is even, it suffices to show that $2$ is even.
  By \textsc{even-nonzero}, to prove $2$ is even, it suffices to show that $0$ is even.
  Finally, we know that $0$ is even by \textsc{even-zero}.\footnote{%
    This proof is written in a \emph{backwards} style where we take the claim and refine it to an axiom, in this case, the \textsc{even-zero} rule.
    Alternatively, we could have written the proof in a \emph{forwards} style where we start with an the axiom that $0$ is even and use the deductive rules to eventually conclude $4$ is even.
  }
\end{proof}

\subsubsection{Authoring Logics}

\begin{figure*}
\begin{haskellcode}
-- Orca provides various "mk" functions, eg mkRule,
--  as well as asParser, SymV, RelV (all used below)

-- * Sorts and Relations
data Nat = N' Id | O | S Nat
  deriving Generic

newtype EvenRel = EvenRel Nat
  deriving Generic

instance SymV Nat
instance RelV EvenRel

-- * Inference Rules
zeroEven, sucEven :: InferRule

zeroEven = mkInferRule0 "zeroEven" [] [] concl
  where
    --------------------
    concl = EvenRel O

sucEven = mkInferRule1 "sucEven" [] [] concl prem
  where
    prem = EvenRel (N' "n")
    -----------------------------------------
    concl = EvenRel (S (S (N' "n")))

-- * Parser and pretty-printer code omitted
-- pNat, ppNat, pEvenRel, ppEvenRel

-- * Logic
evenLogic :: Logic
evenLogic = mkLogic "Even" "EvenRel" builder
  where
    builder = do
      mkSort "Nat" $ mkSymTranslator (asParserM pNat) ppNat
      mkInferSys "EvenRel" ["Nat"] (mkRelTranslator
        (asParserM pEvenRel) ppEvenRel ppEvenRel) [zeroEven, sucEven]

db :: LogicDb
db = [ mkEntry evenLogic ]
  where
    mkEntry l = (logicName l, l)

\end{haskellcode}
\caption{Example even logic declared in \orca.}
\label{fig:even-orca}
\end{figure*}

\orca's embedded domain-specific language gives instructors the power and flexibility of the full Haskell programming language to specify their logic in a type-safe manner.
\autoref{fig:even-orca} gives the declaration of the evenness logic in \orca\footnote{%
  The figure elides the parser and pretty-printer for natural numbers and the evenness relation which are standard.
}.
The instructor declares the sorts and relations of their logic using standard Haskell algebraic datatypes.
We declare a datatype for natural numbers, \haski{Nat}, with a standard peano-style unary representation---a \haski{Nat} is either zero, written \haski{O}, or the successor of a natural number \haski{n}, written \haski{S n}---and translate between numeric literals and this representation during parsing and pretty-printing.
The \haski{N'} constructor of \haski{Nat} acts as an \orca \emph{meta-variable}, standing in for any \haski{Nat} in the definition of the logic's inference rules.
We also declare a datatype to represent the even relation, \haski{EvenRel}, with a sole constructor that takes a \haski{Nat} as an argument.
The eDSL uses Haskell's \haski{generic deriving} mechanism to automatically tie these datatypes to Core \orca's internal representations for symbols and relations.
All subsequent definitions use these datatypes and Haskell's type system ensures that we do not mix up sorts and relations in erroneous ways.

We declare the inference rules for the system using these datatypes as well.
With some \orca library functions (\haski{mkInferRule0} and \haski{mkInferRule1} for making inference rules with zero and one premise, respectively) and some stylized syntax, we are able to declare the \textsc{even-zero} and \textsc{even-nonzero} rules in a manner similar to how we would write them down on paper (\haski{zeroEven} and \haski{sucEven} in \autoref{fig:even-orca}, respectively).

Finally, we bundle all these definitions together in an \orca \haski{Logic} which ties together the declared sorts, relations, and relation actions together.
With this simple logic, we require no automation and thus no custom relation actions are necessary.
The \haski{mkInferSys} helper function, in particular, registers a new relation, in this case \haski{EvenRel}, and assigns it the basic inference action.

\subsubsection{Relation Actions}

\begin{figure*}
  \begin{haskellcode}
-- RId, Prop previously defined...
  
data Ctx
  = G' Id
  | GNil
  | GCons RId Prop Ctx
  deriving (Show, Eq, Generic)

data IsInRel = IsInRel RId Prop Ctx
  deriving (Show, Eq, Generic)

lookupAction :: PremiseAction
lookupAction = PremiseAction action
  where
    action (fromRel -> Just (IsInRel (RId x) p g)) =
      case find x g of
        Just p' ->
          if p == p' then
            dischargeGoal
          else
            failwith $ DoesNotApplyErr "Proposition does not match"
        Nothing -> failwith $ DoesNotApplyErr ("Hypothesis not found: " ++ x)
    action r = failwith $ DoesNotApplyErr (show r)
    find _ GNil = Nothing
    find x (GCons (RId y) p g) =
      if x == y then pure p else find x g
    find x (GCons (R' _) _ g) = find x g
  \end{haskellcode}
  \caption{Context lookup action in \orca.}
  \label{fig:orca-example-action}
\end{figure*}

More complicated sorts of relation actions can be defined within the DSL.
For example, in a first-order logic, one might encode a \emph{context} $\Gamma$ which maps variables to propositions assumed to be provable.
\[
\Gamma = x_1{:}p_1, \ldots, x_k{:}p_k.
\]
In elementary presentations of first-order logic, one might include a rule for proving a propositional variable whenever the context contains exactly that variable:
\[
\inferrule[var-exact]
  {}
  {x{:}p \vdash p}
\]
along with rules for manipulating the context so that a proof can \emph{contract} a context to a single entry of interest:
\[
\inferrule[ctx-exchange]
  {x_2{:}p_2, x_1{:}p_1, \Gamma \vdash p}
  {x_1{:}p_1, x_2{:}p_2, \Gamma \vdash p}
\qquad
\inferrule[ctx-weaken]
  {x_2{:}p_2, \Gamma \vdash p}
  {x_1{:}p_1, x_2{:}p_2, \Gamma \vdash p}.
\]
Depending on the instructor's pedagogical goals, this particular presentation of context lookup might be too burdensome to the student.
Rather than relying on manually-entered inference rules, \orca allows the instructor to define a non-trivial premise action which performs context lookup automatically.
\autoref{fig:orca-example-action} gives the code for this premise action which is similar in implementation to the standard association list \haski{lookup} function. 
With this action in place, \orca automatically determines whether an appropriate assumption is in the context when the variable rule is invoked.
More generally, our system of relation actions allow the instructor to define proof automation in a fine-grained way that aligns with their intended pedagogy rather than fighting with the proof assistant to get the level of automation that they want.

\subsubsection{Rules for Evaluation}

To implement the rules for substitutive evaluation from \autoref{sec:substitution}, we define a logic that utilizes a notion of \emph{evaluation context}~\cite{felleisen-1992} in order to succinctly describe how to evaluate a sub-expression of a program.

For example, we might write the \textbf{if-True} rule of \autoref{sec:subst-if} on paper as:
\[
  \inferrule[if-true]
    {}
    {\pyi{e1 if True else e2} \longrightarrow \pyi{e1}}
\]
However, the 3-2-1 curriculum allows the student to replace a conditional that appears \emph{anywhere} inside the program, not just programs that themselves are conditionals.
Effectively, we must \emph{search} for such a conditional that has the appropriate shape and \emph{replace} that conditional with a new expression according to the rule, \emph{leaving the containing expression untouched}.

To do this, we break up a program into an outer \emph{evaluation context}, $E$, and a sub-expression of interest $e$.
We can think of this evaluation context $E$ as an expression with a hole that is filled by the expression $e$.
%
%
%
For example, the overall expression \pyi{1 + (0 if True else 100)} can be broken up into an evaluation context $E = \pyi{1 + _}$ and a sub-expression $e = \pyi{(0 if True else 100)}$, allowing us to focus on the conditional for the purposes of applying the conditional rule.
Note that the syntax of possible evaluation contexts $E$ determines the evaluation orders that the overall logic admits.

\begin{figure*}
  \begin{haskellcode}
ifTrue :: InferRule
ifTrue = mkInferRule3 "ifTrue" [] [] concl prem1 prem2 prem3
  where
    prem1 = DecomposeRel (E' "e") (EC' "ctx") (EIf (EBool True) (E' "e1") (E' "e2"))
    prem2 = RecomposeRel (EC' "ctx") (E' "e1") (E' "eout")
    prem3 = StepRel (S' "sig") (E' "eout")
    -------------------------------------------------------------------------------
    concl = StepRel (S' "sig") (E' "e")
    
-- where StepRel is the relation for the 3-2-1 S.N.M., analogous to the "even" logic's EvenRel
  \end{haskellcode}
  \caption{Inference rule declaration of \textbf{if-True} in \orca.}
  \label{fig:orca-if-true}
\end{figure*}

With this machinery, we can rewrite the conditional rule in a form amendable to \orca:
\[
  \inferrule[if-true-ctx]
    {}
    { E[\pyi{e1 if True else e2}] \longrightarrow E[\pyi{e1}] }.
\]
This version of the rule allows us to replace a true conditional with its if-branch \emph{anywhere} such a conditional occurs in our program.
Implicit in this rule is the \emph{decomposition} of the initial program into its evaluation context and sub-expression and the \emph{re-composition} of a context and new sub-expression into a new program.
The \orca implementation of this rule found in~\autoref{fig:orca-if-true} makes this explicit through two relations \haski{DecomposeRel} \haski{RecomposeRel} whose associated relation actions perform decomposition and re-composition of the current goal expression automatically.
This code is equivalent to the final version of the rule as realized in \orca:
\[
  \inferrule[if-true-orca]
    {%
      e = E[\pyi{e1 if True else e2}] \\
      e' = E[e1] \\
    }
    { e \longrightarrow e' }
\]

\subsubsection{Block-based Front-end}

To ease adoption in undergraduate classrooms, \orca is a web application immediately accessible to undergraduates through any modern web browser.
The back-end server application is customized by the instructor to serve a particular set of logics to students written in the eDSL described above.
The client-side UI allows students to write propositions and author proofs in those logics.

Clients query the server for a particular logic and if the server supports that logic, sends the client the information necessary for students to author proofs in that logic:
\begin{itemize}
  \item The sorts and relations of the logic and
  \item The inference rules for the logic.
\end{itemize}
The client sends a serialized version of the proof back to the \orca server and server responds with the results of checking the proof.
Critically, this includes the state of the proof as it evolves through each of the provided proof steps.
As such, the client is completely \emph{stateless}---it sends complete proofs to the server and receives complete results back to be consumed by the user.

\begin{figure}
  \includegraphics[width=\textwidth]{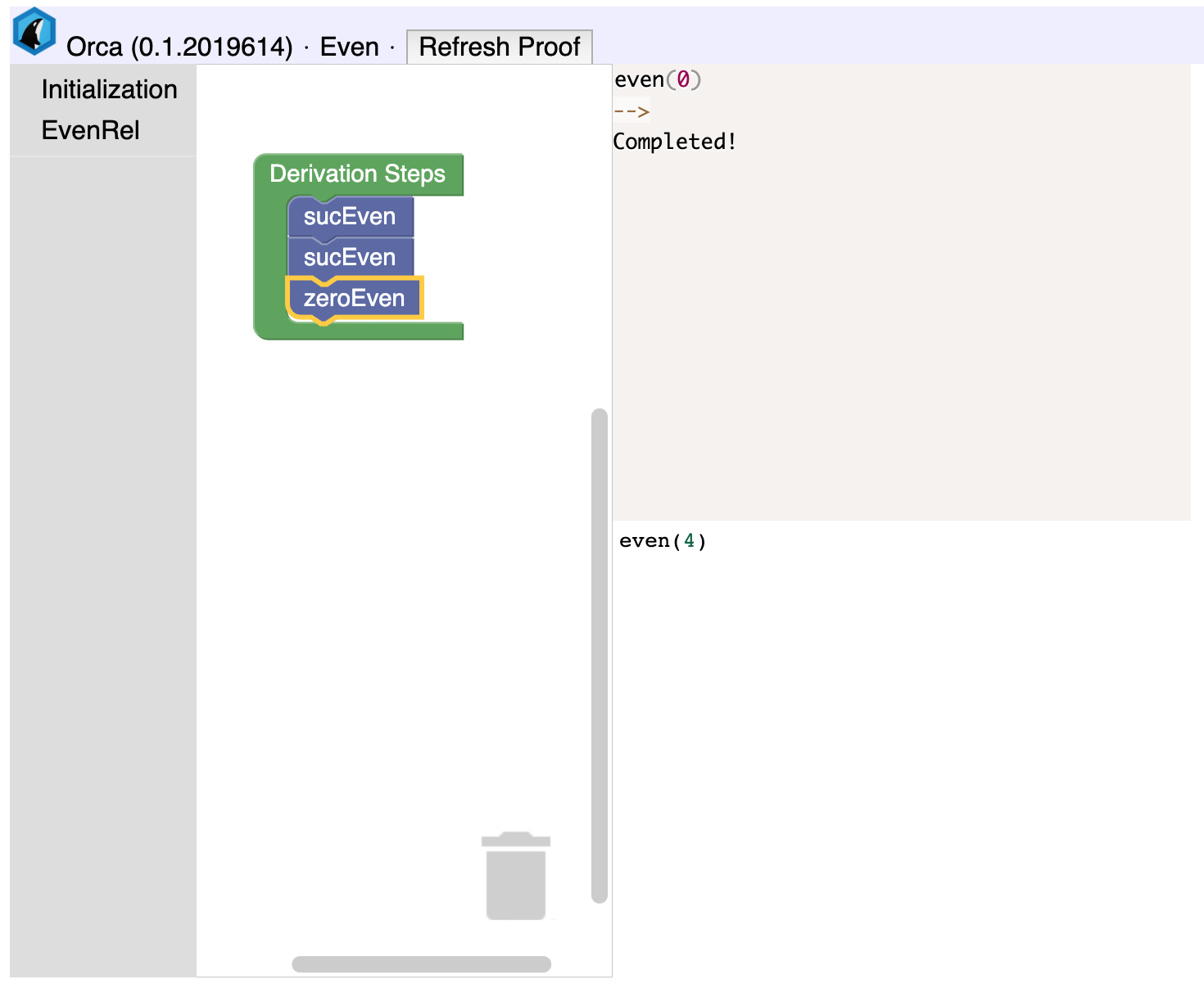}
  \caption{$even(4)$ as realized in \orca}
  \label{fig:orca-even-ss}
\end{figure}

\orca's current front-end is a block-based front-end built on the Blockly framework~\cite{pasternak-2017}.
\autoref{fig:orca-even-ss} gives an example of the evenness logic rendered in \orca along with the proof that four is even.
The initial proof goal is given by the user in the bottom-right pane of the \orca application.
Sequences of blocks in the framework correspond to proof steps where each individual block corresponds to an inference rule applied to the proof state at that point in the proof.
The upper-right \emph{status pane} displays the before-and-after state of the proof at the highlighted step.
In \autoref{fig:orca-even-ss}, the final step of the proof, \textsc{zeroEven} is highlighted, and the status pane indicates that the state of the proof before applying the rule is $even(0)$ (\ie, the user must prove $even(0)$) and the proof is complete after applying the rule.
Orca can support more interesting structures,
such as proof by cases, or the introduction of an induction hypothesis,
to support logics that are best framed in the classic mathematical conjecture/proof structure.

\section{Integrating CS1 and DM}
\label{sec:integrating}

The CS1 and Discrete Mathematics curricula described above
do not currently exist in the same environment;
Haverford's CS1 is paired with a discrete math course that does not feature \orca (or Coq),
and Grinnell's discrete math is paired with a CS introductory course without explicit code reasoning.

The relationships between proof-construction, reasoning about computation, and coding
could be brought to the fore by combining the curricula at a single institution.
This would let students see the deep similarities in the use of formal reasoning
in CS and mathematics,
and also recognize and appreciate the different structures.

The substitutional execution of Figure~\ref{fig:ints-code} is,
of course, an illustration of direct proof technique, in a different guise.
We can emphasize this for students by prefixing
figures~\ref{fig:ints-arith} or \ref{fig:ints-arith2}
with ``Claim: x = 42'' and, ending with ``Q.E.D.'',
to provide familiar mathematical structure.

Proof-by-cases is equivalent to our applying the
if-irrelevant rule of Section~\ref{sec:subst-if}
``in reverse'', \ie, to introduce, rather than remove, an \pyi{if}.
For example,
for Section~\ref{sec:discrete-mathematics}'s proof that
$( \forall \oi{b}.\,\oi{andb b false} = \oi{false} )$,
this reverse-application can be used to posit \pyi{b==True},
turning
\pyi{(b and False)},
into
\pyi{((b and False) if b==True else (b and False))};
after further simplification, an ordinary application of if-irrelevant
turns \pyi{((False) if b==True else (False))}
into \pyi{False}, completing the proof.
Conversely, the application of our if-irrelevant rule in Figure~\ref{fig:power-pf1}
could instead have been handled by considering cases that \pyi{y == 1} could be true or false.

The overall structure of our correctness proofs for recursive functions is,
of course, that of inductive proof.
Figure~\ref{fig:power-pf1} is a proof of the inductive hypothesis
$forall \oi{x,y}, \oi{x} \in \mathbb{R} \land \oi{y} \in \mathbb{Z}^+,
\oi{power(x,y)}=x^y$,
by induction on $\oi{y}$.
Illustrating this one fact in both the substitution notional machine
and the classic structure of mathematical proof,
may help students see the connection.
This illustration is particularly relevant if the instructor for CS1 plans to move
into correctness proofs for functions whose
expected results are described via properties rather than equations,
as per Section \ref{sec:symbolic-additional}
(where we switch to the classic mathematical form in the 3-2-1 curriculum).

Note that it is often possible to convert induction into the substitution notional machine,
though we do not bring this up in class,
as we expect it would create more confusion than clarity.
But, for the curious,
note that we could express Section~\ref{sec:dm-induction}'s
claim ``Append Preserves Length''
as a statement about a function \pyi{applen} that returns
\pyi{len(append(l1, l2))},
for which the post-condition specifies that it returns \pyi{len(l1) + len(l2)}.
(Note that this requires us to use some substitutions ``in reverse'',
e.g. converting \pyi{0} into \pyi{0 if False else 1+len(tail(l1))},
and then eventually recognizing the body of the \pyi{len} function and
converting to \pyi{len(l1)} with the reverse of the name-to-body rule.
As stated above,
this example is not something we usually want to put into the substitution form.)

\TRvsCCSC{
\section{Experimental Validation}
\label{sec:experiments}

We hope to design experiments to measure many of the following questions,
though some may not be answerable within a study of reasonable size.
Some have already been studied in other contexts,
\eg, in~\citeN{fkt:sigcse:2017sma,tfk:sigcse:2018notional}
and we may refine our questions to let us align with their prior experiments,
and thus improve comparability.

\subsection{Overall Value of Substitution and Deduction in CS1}

We would like to know:
when compared to an essentially similar CS1 that replaces
exercises in deduction and proof with additional coding practice,
does this material
\begin{description}
  \item[RQ V-1] improve student ability to perform formal deduction on code?
  \item[RQ V-2] decrease or increase programming and logical errors in their code?
  \item[RQ V-3] increase student appreciation of the relevance of mathematical elements of other courses?
  \item[RQ V-4] increase interest in subsequent computer science courses (or, specifically, formal courses)?
  \item[RQ V-5] increase success rates in later more-formal courses of the major?
  \item[RQ V-6] increase success rates in later courses that rely on detailed models of execution, \eg concurrent programming?
  \item[RQ V-7] reduce outcome disparities created by uneven mathematical preparation prior to entrance to college?
\end{description}

We are also curious about the degree to which the answers to these questions is sensitive to the content of
the discrete math course,
\eg, whether or not it contains any material about proofs involving code,
or whether it uses the 3-2-1 SNM and/or \orca specifically.

We are also curious about early approaches to predicting success in later courses
(\ie, RQ V-5 and RQ V-6)
at the end of CS1.
For example, we might try to probe the ability to generalize important lessons about deductive proof
by asking students to find the mistake in the classic ``1+1=1'' false proof,
perhaps on both a pre-test and post-test for CS1.
Or, we might probe their ability to think about execution by asking questions
about situations in which actual execution order can impact the program result,
\eg simplified examples of concurrent programming.

\subsection{Pedagogy Details}

Earlier sections of this document note a number of pedagogical choices.
The overall value of the curriculum could well be very sensitive to these choices.
Among the questions we'd like to answer:
\begin{description}
  \item[RQ D-1] 
                At which point(s) in our sequence
                from substitution, to deduction of an abstract property, to verification,
                do students become confused? (Essentially, RQ1 of~\citeN{tfk:sigcse:2018notional}.)
  \item[RQ D-2] Are students confused by the very idea of flexibility in execution order,
                and, if so, are they helped by considering analogies to other computational tasks
                that can be solved in several ways, such as sorting?
  \item[RQ D-3] How much are students helped by guidance about which rule to use first,
                \eg by showing left-to-right applicative-order evaluation,
                or providing one or more \orca option(s) to do so automatically,
                rather than encouraging them to explore different order choices?                
  \item[RQ D-4] Is the name-to-body rule of Section~\ref{sec:subst-body} more, or less, confusing
                than the statement-centric system we had been using (or the version with $\lambda$,
                of Sections~\ref{sec:subst-details-funcdef} and~\ref{sec:subst-details-localdef})?
  \item[RQ D-5] For cases in which a substitution is legal only within a specific context,
                such as the simplification of \pyi{y>1} in Figure~\ref{fig:power-pf1},
                to what degree should students be forced to consider context,
                to support their later work with mathematical proofs?
  \item[RQ D-6] Do Section~\ref{sec:subst-details}'s approaches to foregrounding scope rules for variables
                help students with this issue? (We hope to employ the relevant examples used in \cite{fkt:sigcse:2017sma,tfk:sigcse:2018notional}.)
\end{description}

\subsection{Tool Support via Orca}
In evaluating \orca{} and its associated pedagogy, we are interested in answering the following research questions:
\begin{description}
  \item[RQ O-1] Does \orca{} enhance student ability to understand and perform deductive reasoning?
  \item[RQ O-2] Does \orca{} enhance student confidence in learning mathematics?
  \item[RQ O-3] Does \orca{} enhance student understanding of the connection between mathematics and computer science?
\end{description}

We are also interested in the degree to which
the answers to our ``overall value'' and ``pedagogy details'' questions
are sensitive to the use of \orca (when compared to manual application of the same substitution rules).
In particular, we are interested in assessing \orca{}'s ability to ``level the playing field'' for those students that come to the our introductory courses with lesser mathematical backgrounds or less confidence in their mathematical abilities.

}{}

\section{Conclusion}
\label{sec:conclusion}

A number of institutions have used substitution-based models of computation
in introductory computer science for many years.
Independently, proofs involving code
have been used to connect discrete mathematics courses to other elements of a computer science major.
We have explored the use of the deep connections among substitution, deduction, and proof,
to enhance and deeply interconnect these two critical courses that begin study of computer science.

Our curricula for CS1 and discrete math are independent but mutually-supporting.
They also offer several choices for the curricular dividing line between the two:
as presented here, all of the logical elements of an inductive proof can be presented in CS1,
within the framework of verification of software by execution for symbolic inputs.
Our CS1 curriculum can also connect more deeply into mathematical structure and terminology,
by re-framing that verification,
or exploring verifications of functions with more interesting specifications
or imperative codes that are verified with multi-obligation proof structures.
Alternatively, CS1 could entirely omit verification, or even symbolic execution,
by selective application of the rules of the 3-2-1 substitution notional machine.
A discrete mathematics course can then pick up this sequence,
with the shared \orca tool emphasizing continuity of tools and concepts,
and slight variations of the \orca logic emphasizing differences in structure and vocabulary.

We believe our new variants of the 3-2-1 substitution rules, applied with \orca,
will provide a minimally-difficult pathway for new programmers connect program execution to proof,
and that this connection is best made by the use of coordinated tools and examples
that appear in both introductory programming and discrete mathematics.
We hope to test these beliefs in the coming years.


\begin{acks}

We would like to thank Shriram Krishnamurthi, Kathi Fisler, and their students at Brown University,
for pushing us to think more deeply about how to minimize confusion for new CS students.
Their questions and thoughts have helped us identify just what is, or is not, required to establish our
connection between code and proof.
We would also like to thank our students whom have worked on \orca thus far: Myles Becker, Elise Brod, Jerry Chen, Addison Gould, Medha Gopalaswamy, Yash Gupta, Hadley Luker, Eli Most, Lukas Resch, Sooji Son, and Prabir Pradhan.

We are also deeply appreciative of our students,
who have survived, and even thrived, in earlier steps of our iterative confusion-minimization process.
Their comments and feedback, often offered with impressive calm and maturity,
have greatly improved our curricula and teaching.

This material is based upon work supported by the \grantsponsor{nsf}{National Science Foundation}{https://nsf.org} under Grant No.~\grantnum{nsf}{1651817}.
Any opinions, findings, and conclusions or recommendations expressed in this material are those of the author and do not necessarily reflect the views of the National Science Foundation.
\end{acks}

%
\bibliographystyle{plainnat}   
\bibliography{SymExecProof}

%

\end{document}